\documentclass[11pt]{article}
\usepackage{amssymb,amsmath,amsthm,latexsym,natbib,amsbsy}
\usepackage{graphicx}
\usepackage{float}
\usepackage{caption}
\usepackage{subcaption}
\usepackage{epsfig}
\usepackage{epstopdf}  
\usepackage{color}
\usepackage[font= {small,it} ]{caption}
\usepackage{url}
\usepackage[titletoc]{appendix}
\usepackage{xcolor}  

\usepackage[margin= 1in]{geometry}

\DeclareMathOperator*{\argmax}{arg\,max}

\newcommand*\samethanks[1][\value{footnote}]{\footnotemark[#1]}

\newtheorem{lemma}{Lemma}

\title{A Bayesian Generalized CAR Model for Correlated Signal Detection}
\author{D. Andrew Brown\thanks{Corresponding author, Department of Mathematical Sciences, Clemson University, Clemson, SC 29634, Email: ab7@clemson.edu} \and  Gauri S. Datta\thanks{Department of Statistics, University of Georgia, Athens, GA 30602} \and  Nicole A. Lazar\samethanks}

\begin{document}
\maketitle

\thispagestyle{empty}

\begin{abstract}
\noindent Over the last decade, large-scale multiple testing has found itself at the forefront of modern data analysis. In many applications data are correlated, so that the observed test statistic used for detecting a non-null case, or signal, at each location in a dataset carries some information about the chances of a true signal at other locations. Brown, Lazar, Datta, Jang, and McDowell (2014) proposed in the neuroimaging context a Bayesian multiple testing model that accounts for the dependence of each volume element on the behavior of its neighbors through a conditional autoregressive (CAR) model. Here, we propose a generalized CAR model that allows for inclusion of points with no neighbors at all, something that is not possible under conventional CAR models. We consider also neighborhoods based on criteria other than physical location, such as genetic pathways in microarray determined from existing biological knowledge. This generalization provides a unified framework for the simultaneous modeling of dependent and independent cases, resulting in stronger Bayesian learning in the posterior and increased precision in the estimates of interesting signals. We justify the selected prior distribution and prove that the resulting posterior distribution is proper. We illustrate the effectiveness and applicability of our proposed model by using it to analyze both simulated and real microarray data in which the genes exhibit dependence that is determined by physical adjacency on a chromosome or predefined gene pathways.\\

\end{abstract}

\newpage

\thispagestyle{empty}

\section{Introduction}

The analysis of high throughput data presents many challenges to researchers across a variety of disciplines. Many of the problems that must be dealt with are ubiquitous in the sciences but are exacerbated when the datasets are massive in size. Often, the goal is to detect the presence or absence of a signal over an extremely large number of cases, creating a massive multiple testing problem. Prior to the last two decades, most multiple testing procedures were constructed to control an overall error rate for a relatively small number of simultaneous tests \citep*{Efron10}. The advent of high throughput technology quickly revealed, however, that classical procedures can be inappropriate in the presence of thousands of simultaneous tests \citep*{Benjamini2010}.

Suppose our data consist of $J$ cases, each of which arises independently from a normal distribution with case-specific mean, $y_j \sim N(\theta_j, \sigma^2), ~j=1,\ldots,J$. For example, in gene microarray or next-generation RNA sequencing, $y_j$ may be the test statistic quantifying the differential expression of gene $j$ between cancerous and healthy tissue \citep*{EfronTib02, LiTibshirani11}. In functional magnetic resonance imaging (fMRI), signals are observed over time at thousands of points in the brain and a test statistic $y_j$ is calculated at each point $j$ to summarize the observed difference in that area's signal between some stimulus condition versus a baseline \citep*{FristonEtAl95}. Frequently, a question of interest is whether or not $\theta_j = 0$ at each $j$; i.e., a hypothesis test is conducted at each of thousands of locations to determine which of the $\theta_j$ are non-zero, indicating an interesting signal. The problem is to find a statistical procedure which corrects for multiple testing but does not sacrifice too much sensitivity for the sake of preventing false positives.

A similar issue arises in variable selection, where one is interested in determining which (of possibly many) variables contribute in a meaningful way to the observed response. A Bayesian approach is to assume the coefficient corresponding to each variable, $\beta_j$, belongs to one of two distinct classes, the null class in which $\beta_j \sim N(0, \sigma^2)$, or the non-null class, $\beta_j \sim N(\theta_j, \sigma^2), ~\theta_j \neq 0$ \citep*[e.g.,][]{MitchellBeauchamp88, GeorgeMcCulloch93, ScottBerger06, Efron08}. Each $\beta_j$ is assigned an {\em a priori} probability $p$ of belonging to the null class. This value may be regarded as the proportion of all possible cases which would be null so that $p \approx 1$ models very sparse signals, while $p \approx 0$ models abundant signals. This mixing proportion, $p$, is usually unknown, but it can be assigned a prior distribution to reflect the researcher's beliefs about the level of sparsity in the data, or it can be estimated via empirical Bayes. The posterior {\em inclusion probabilities} \citep*{BarbieriBerger04} can subsequently be estimated from the posterior distribution. Placing a Beta prior on $p$ induces a multiplicity adjustment in that the model automatically penalizes for the number of tests in {\em a posteriori} probability statements. \cite*{ScottBerger10} discuss this issue and the conditions under which multiplicity correction can be induced.

Much of the work thus far developed is based on the assumption of independent hypothesis tests. This assumption is untenable in many applications. For example, nontrivial dependence structures are known to exist in neuroimaging data \citep*{LeeEtAl14, ZhangEtAl14}, syndromic surveillance \citep*{BanksEtAl09}, gene microarray \citep*{ZhaoEtAl14}, and RNA sequencing \citep*[RNA-seq;][]{LoveEtAl14}. Correlation can cause the null distribution of the observed test statistics to be over- or under-dispersed relative to the theoretical null under independence. Consequently, either too few or too many test statistics may be declared significant. The deleterious impact correlation can have on empirical Bayes methods and false discovery rate control \citep*[FDR;][]{BenjaminiHochberg95} was investigated in \cite*{QiuEtAl05}. \cite*{Efron07} focused on the effects of dependence on the distribution of test statistics. While some work has been done on incorporating known dependence structure into Bayesian models for identification of interesting cases \citep*[e.g.,][]{SmithFahr07, LiZhang10, StingoEtAl11, LeeEtAl14, ZhaoEtAl14, ZhangEtAl14}, it has been limited, particularly with respect to exploring the multiple testing adjustments incurred through data-dependent estimation of inclusion probabilities.

Many datasets include isolated observations. For example, genes in microarray data share common pathways, but many genes are in no pathway at all. It is important to include as many cases as possible when evaluating the posterior distribution. A standard CAR structure assumes every observation has at least one neighbor, so that one is forced to either use an inappropriate neighborhood structure or exclude isolated points altogether. In the current paper, we extend a model proposed by \cite*{BrownEtAl14} for Bayesian multiple hypothesis testing and discuss more general neighborhood structures to provide a unified treatment of cases with at least one neighbor and with no neighbor. We use a less restrictive improper prior distribution for the variance components and establish the propriety of the posterior distribution of our model, ensuring that inferences are valid. In addition to allowing the newly proposed model to identify isolated non-null cases, we demonstrate that the inclusion of isolated points results in stronger Bayesian learning and improved estimation of the signal strengths of the selected cases.

We briefly motivate a common model used in Bayesian signal detection in Section 2. This leads to our proposed extension of the so-called spike-and-slab prior to accommodate local dependence. We prove the propriety of the posterior distribution of our proposed model and discuss computation. In Section 3, we simulate correlated microarray data to study our model's performance against a prior assuming independence and assuming a conventional CAR structure. We also compare it against results obtained from the significance analysis for microarrays (SAM) procedure \citep*{EfronEtAl01, TusherEtAl01, Efron10}. We apply our procedure to two real gene microarray datasets in Section 4, one using dependence determined by adjacency on a chromosome, and the other with gene pathways defining the neighborhoods. Section 5 concludes with a discussion of the results.

\section{Methods}\label{sec:Methods}

\subsection{Mixture Priors for Multiplicity Adjustment}
To facilitate a Bayesian multiple testing correction, we postulate a ``two groups model" \citep{Efron08}. That is, we assume two possible cases for each of the observed test quantities, reflected in a Bayesian model through the prior on $\theta_j$,
\begin{equation}
    \label{eqn:spikeSlab}
    \pi(\theta_j \mid p, ~\tau^2) = p\delta_0(\theta_j) + (1-p)\varphi_{0,\tau^2}(\theta_j) ~~j=1, \ldots, J,
\end{equation}
where $\delta_0(\cdot)$ is the Dirac delta spike at zero and $\varphi_{0,\tau^2}(\cdot)$ is the Gaussian density function with mean zero and variance $\tau^2$. So-called ``spike-and-slab" or ``spike-and-bell" priors of this form are standard in the Bayesian variable selection framework. They were introduced by \cite{MitchellBeauchamp88} for variable selection in linear regression. The mixture model was used in \cite{Geweke96}, who provided a procedure for selecting models subject to order constraints among the variables included in each model. A similar approach was taken in \cite{GeorgeMcCulloch93}, who treated each regression coefficient as arising from a mixture of two continuous distributions with different variances for stochastic search variable selection. Literature on Bayesian variable selection was reviewed in \cite{ClydeGeorge04} and \cite{OHaraSillanpaa09}. \cite{ScottBerger06} explored Model (\ref{eqn:spikeSlab}) and ways in which it induces multiplicity correction, along with graphical displays and decision rules for subsequent inferences.

By allowing $p$ to be determined by the data, the joint posterior distributions obtained from spike-and-slab priors adapt to the number of tests, resulting in the posterior inclusion probabilities, $p_j := P(\theta_j \neq 0 \mid \mathbf{y})$, being penalized to account for the multiple tests \citep{ScottBerger06, ScottBerger10}. Specifically, \cite{ScottBerger06} used a $\text{Beta}(\alpha, 1)$ prior density on $p$, where $\alpha$ is a specified value, and $\pi(\tau^2, ~\sigma^2) = (\tau^2 + \sigma^2)^{-2}$ as a prior density for the variance components. Under this model, \citet[][Lemma 3]{ScottBerger06} showed that
\begin{equation}\label{eqn:SBInclProbExp}
    p_j = 1 - E\left[\left(1 + \frac{1-p}{p}\sqrt{\frac{\sigma^2}{\sigma^2 + \tau^2}}\exp\left(\frac{y_j^2\tau^2}{2\sigma^2(\sigma^2 + \tau^2)}\right)\right)^{-1}\right],
\end{equation}
where the expectation is taken with respect to the joint posterior distribution of $p$, $\sigma^2$, and $\tau^2$.\\

\subsection{Incorporation of Local Dependence}\label{sub:propModel}
Posterior inference can be sharpened if we exploit correlation among potential predictors in the search for interesting signals. \cite*{BrownEtAl14} proposed allowing the continuous component of (\ref{eqn:spikeSlab}) to share information across observations by reparameterizing $\theta_j$ as $\gamma_j\mu_j$, where $\gamma_j \stackrel{\text{iid}}{\sim}\text{Bern}(1-p)$ and $\mu_j$ is Gaussian, so that $\mathbf{y} \sim N(\boldsymbol{\Gamma\mu}, \sigma^2\mathbf{I})$, where $\mathbf{y} = (y_1, \ldots, y_J)^T, ~\boldsymbol{\Gamma} = \text{diag}\{\gamma_i, ~i=1, \ldots, J\},$ and $\boldsymbol{\mu} = (\mu_1, \ldots, \mu_J)^T$. The lattice structure of datasets such as those arising from fMRI and gene microarray, makes a conditional autoregressive model \citep*[CAR;][]{Besag74} a natural choice for incorporating local dependence into the prior on $\boldsymbol{\mu}$. Since the potential non-null signals are expected to be as much positive as negative, {\em a priori} it is reasonable to assume such signals have zero means. Thus, we consider prior distributions of the form $\mu_j \mid \boldsymbol{\mu}_{_{(-j)}} \sim N\left(\sum_{i=1}^J c_{ji}\mu_i, \tau^2_j\right), ~j= 1,\ldots, J$, where $\boldsymbol{\mu}_{_{(-j)}} = (\mu_1, \ldots, \mu_{j-1}, \mu_{j+1}, \ldots, \mu_J)^T$, $c_{jj} = 0$, and $c_{ji} = 0$ except when cases $j$ and $i$ are neighbors. The intrinsic autoregressive model \citep*[IAR;][]{BesagEtAl91} emerges by taking $c_{ji} = w_{ji}/w_j.$ and $\tau_j^2 = \tau^2/w_j.$, where $w_{ji} \neq 0$ if and only if sites $j$ and $i$ are neighbors and $w_j. = \sum_{i=1}^J w_{ji}$. Under an IAR model for $\boldsymbol{\mu}$, the prior density is given by
\begin{equation}\label{eqn:iarDens}
\pi(\boldsymbol{\mu} \mid \tau^2) \propto \exp\left(-\frac{1}{2\tau^2}\boldsymbol{\mu}^T(\mathbf{D}_w - \mathbf{W})\boldsymbol{\mu}\right),
\end{equation}
where $\mathbf{D}_w = \mbox{diag}\{w_j., ~j= 1, \ldots, J\}$ and $\mathbf{W} = \{w_{ji}\}_{j,i=1}^J$. Note that $(\mathbf{D}_w - \mathbf{W})\mathbf{1} = \mathbf{0}$ so that the precision matrix has a nontrvial nullspace and hence the IAR is improper. However, a ``propriety parameter", $\rho$, can be used in the conditional distributions such that $\mu_j \mid \boldsymbol{\mu}_{_{(-j)}} \sim N(\rho\sum_{i = 1}^Jw_{ji}\mu_i/w_j., ~\tau^2/w_j.)$ with precision matrix $\tau^{-2}(\mathbf{D}_w - \rho\mathbf{W})$. In general, the precision matrix will be nonsingular if $\lambda_1^{-1} < \rho < \lambda_J^{-1}$, where $\lambda_1 < 0$ and $\lambda_J > 0$ are the smallest and largest eigenvalues of $\mathbf{D}_w^{-1/2}\mathbf{W}\mathbf{D}_w^{-1/2}$, respectively \citep*{BanerjeeEtAl04}.

Any data that have a lattice structure with known or suspected correlations occurring along predefined networks can be modeled with a CAR model. For instance, genes in microarray are known to express themselves in clusters along a chromosome \citep*[e.g.,][]{XiaoEtAl09}, or to behave in concert along specific gene pathways \citep*{SubramanianEtAl05}. Thus, neighborhoods can be defined in terms of adjacency on a chromosome or based on genes sharing certain predefined pathways determined from prior knowledge. Care should be taken in defining neighborhoods, though, as microarray / RNA-seq datasets often include genes that are not members of any known pathway and thus are isolated. Including isolated points in the IAR induces zero rows in the precision matrix, a problem that cannot be fixed with a propriety parameter. In response, we adjust the neighborhood weights to allow for inclusion of the isolated points while avoiding a singular precision matrix.

We modify the usual IAR model by defining the neighborhood weights about $\mu_j$ to be $c_{ji} = w_{ji}/(d + w_j.)$ with conditional variance $\tau^2/(d + w_j.)$, where $d > 0$. The consequent precision matrix is $\tau^{-2}(\mathbf{D}_w + d\mathbf{I} - \mathbf{W})$. Then $\mathbf{x}^T(\mathbf{D}_w + d\mathbf{I} - \mathbf{W})\mathbf{x} = \sum_{i=1}^J d x_i^2 + (1/2)\sum_i \sum_j w_{ij}(x_i - x_j)^2 \geq 0$, with strict inequality for $\mathbf{x} \neq \mathbf{0}$. Thus, with $d> 0$, we are able to include isolated points in the model while maintaining the propriety of the distribution. If we take $d = 1$, then for any isolated point $j^{\prime}$, $w_{j^{\prime}}. = 0$ so that $E(\mu_{j^{\prime}} \mid \boldsymbol{\mu}_{_{(-j^{\prime})}}) = 0$, $Var(\mu_{j^{\prime}} \mid \boldsymbol{\mu}_{_{(-j^{\prime})}}) = \tau^2$, and $\mu_{j^{\prime}} \mid \tau^2 \sim N(0, \tau^2)$ independently of other points. Hence, we can facilitate conditional independence of $\mu_j$ while allowing all $J$ points to share information about plausible values of the hypervariance through the prior distribution on the variance components. Taking this view, we can express the traditional IAR model as a special case in which every point in a dataset has at least one neighbor and $d = 0$.

Let the joint density of the data and parameters be given by $f(\mathbf{y}, \boldsymbol{\psi}, \boldsymbol{\mu}, p, \boldsymbol{\gamma})$, where $\boldsymbol{\psi}$ contains nuisance parameters modeled in the prior distribution. When $\gamma_j = 0$ for all $j$, $\boldsymbol{\mu}$ does not appear in the resulting likelihood and thus is Bayesianly unidentified \citep*{GelfandSahu99, EberlyCarlin00}. This means that $\int_{\boldsymbol{\mu}}f(\mathbf{y}, \boldsymbol{\psi} \mid \boldsymbol{\gamma} = \mathbf{0}, \boldsymbol{\mu})\pi(\boldsymbol{\mu})d\boldsymbol{\mu} \equiv \int_{\boldsymbol{\mu}}f(\mathbf{y}, \boldsymbol{\psi} \mid \boldsymbol{\gamma} = \mathbf{0})\pi(\boldsymbol{\mu})d\boldsymbol{\mu} = f(\mathbf{y}, \boldsymbol{\psi} \mid \boldsymbol{\gamma} = \mathbf{0})\int_{\boldsymbol{\mu}}\pi(\boldsymbol{\mu})d\boldsymbol{\mu}$, so that we must have a proper prior on $\boldsymbol{\mu}$ for the posterior distribution to be proper, making the inclusion of $\rho$ necessary when $d = 0$. See \citet*[][Chapter 4]{McLachlanPeel00} for further discussion of prior distributions in finite mixture models.

Usually, there is little direct information available about $\rho$, so estimating it may be difficult. Previous work has shown that appreciable interaction between adjacent points only occurs when $\rho$ is close to its upper bound under the $d = 0$ model \citep*[][]{BanerjeeEtAl04}. To give the data more freedom in determining the spatial association without specific regard for interpretability, we consider the prior $\pi_{\rho}(\rho) \propto I(\lambda_1^{-1} < \rho < \lambda_J^{-1})$. It is important to note that inclusion of $\rho$ is still possible when $d > 0$, provided that $\rho$ is bounded between the reciprocals of the smallest and largest eigenvalues of $(\mathbf{D}_w + d\mathbf{I})^{-1/2}\mathbf{W}(\mathbf{D}_w + d\mathbf{I})^{-1/2}$.

An additional advantage of including $\rho$ in the joint model for $\boldsymbol{\mu}$ is that, under positive spatial association, the posterior distribution becomes insensitive to the choice of $d$ in the neighborhood weights. This is because as $d$ grows, $\rho$ is allowed to increase as well. In other words, if $d_1 < d_2$, then $\lambda_{J,1}^{-1} < \lambda_{J,2}^{-1}$, where $\lambda_{J,i} > 0$ is the maximum eigenvalue of $(\mathbf{D}_w + d_i\mathbf{I})^{-1/2}\mathbf{W}(\mathbf{D}_w + d_i\mathbf{I})^{-1/2}, ~i= 1,2$. A proof of this fact may be found in the Supplementary Material.

We also wish to avoid strong information about either the noise variance, $\sigma^2$, or the hypervariance, $\tau^2$. \cite*{Gelman06} suggested that the priors specified for the variance parameters in hierarchical models may have a disproportionate effect in that they can impose strong restrictions on posterior inference. Conversely, priors used for scale hyperparameters that are intended to be noninformative may, in fact, be {\em too} weak by placing considerable probability on unreasonable extreme values in the posterior. To address this possibility, \cite*{Gelman06} proposed the use of a weakly informative (vague but proper) prior on the scale hyperparameter such as the folded-$t$ distribution. \cite*{ScottBerger06} argued for the use of a joint prior on $\sigma^2$ and $\tau^2$ with density $\pi_{(\tau^2, \sigma^2)}(\tau^2, ~\sigma^2) = (\tau^2 + \sigma^2)^{-2} = (\sigma^2)^{-1}(1 + \tau^2/\sigma^2)^{-2}(\sigma^2)^{-1} \equiv \pi_{\tau^2 \mid \sigma^2}(\tau^2 \mid \sigma^2)\pi_{\sigma^2}(\sigma^2)$ so that a standard improper prior on $\sigma^2$ can be used while scaling $\tau^2$ by $\sigma^2$. The prior on $\tau^2 \mid \sigma^2$ is similar to the prior on $\tau^2$ which results from placing a folded-$t_2$ prior on $\tau$ with scale $\sigma$. Supplementary Figure \ref{fig:SBvHalfT} illustrates a slight difference between the two priors for small values of $\tau^2$ so that the Scott-Berger (SB) prior on $\tau^2$ is slightly less informative than a folded-$t$. However, the SB prior and the folded-$t$-based prior are tail equivalent in the sense that the ratio of the two densities is $O(1)$ as $\tau^2 \rightarrow \infty$. Both priors attempt to balance the freedom of the data to inform about this parameter against the need to protect against unreasonable values. We thus follow the precedent set by \citet*{ScottBerger06} and use the same joint prior distribution on $\tau^2$ and $\sigma^2$.

This leads to the following model:
\begin{equation}
    \begin{aligned}
        y_j \mid \gamma_j, ~\mu_j, ~\sigma^2 &\stackrel{\text{indep}}{\sim} N(\gamma_j\mu_j, \sigma^2); ~~~~\gamma_j \mid p \stackrel{\text{iid}}{\sim} \text{Bern}(1-p), ~~j= 1, \ldots, J\\
        \mu_j \mid \boldsymbol{\mu}_{_{(-j)}}, ~\tau^2, ~\rho &\sim N\left(\sum_{i = 1}^J\frac{\rho w_{ji} \mu_i}{d + w_j.}, ~\frac{\tau^2}{d + w_j.}\right), ~~d \geq 0, ~~j= 1, \ldots, J\\
        p &\sim \text{Beta}(\alpha,1), ~~\alpha \geq 1; ~~~~\rho \sim \text{Unif}(\nu^{-1}_1, ~\nu^{-1}_{J})\\
        \pi_{\tau^2\mid\sigma^2}(\tau^2\mid\sigma^2) &= \left(\frac{1}{\sigma^2}\right)\left(1+\frac{\tau^2}{\sigma^2}\right)^{-2}, ~~\tau^2 > 0; ~~~~\pi_{\sigma^2}(\sigma^2) = \frac{1}{\sigma^2}, ~~\sigma^2 > 0,\\
    \end{aligned}\label{eqn:propModel}
\end{equation}
where $\nu_1$ and $\nu_J$ are the smallest and largest eigenvalues of $(\mathbf{D}_w + d\mathbf{I})^{-1/2}\mathbf{W}(\mathbf{D}_w + d\mathbf{I})^{-1/2}$, respectively. Since we are using an improper prior in (\ref{eqn:propModel}), the posterior density is not guaranteed to be integrable. We provide a proof in the Appendix that the posterior distribution is indeed proper.

The practical effect of $\rho$ having room to increase along with $d$ in our proposed model can be seen through simulation. We simulate a $20\times 20$ array of observations arising from both null and non-null distributions. The activation pattern is created by drawing from an Ising distribution \citep*[e.g.,][]{Higdon94}, $p(\mathbf{x}) \propto \exp(\beta\sum_{i \sim j}I(x_i = x_j)), ~\mathbf{x} \in \{0,1\}^{400}$, with interaction parameter $\beta = 0.35$. The null cases ($x_i = 0$) are drawn from $N(0,1)$ and the non-null ($x_i = 1$) cases are drawn from $N(3.5, 1)$. The binary activation pattern and simulated data array are displayed in Supplementary Figures \ref{fig:ActPattern} and \ref{fig:SimDatArray}. We use these data to estimate the posterior distributions under model (\ref{eqn:propModel}) with $d = 0, ~1, ~\text{and} ~5$. See Subsection 2.3 for implementation details. A descriptive measure of spatial association is Moran's $I$, $I(\mathbf{y}) = n\sum_i\sum_jw_{ij}(y_i - \overline{y})(y_j - \overline{y})/[(\sum_{i\neq j}w_{ij})\sum_i(y_i - \overline{y})^2]$, where values away from zero are evidence of spatial association according to the predefined neighborhood structure \citep*[e.g.,][Sec. 4.1]{BanerjeeEtAl04}. Figure \ref{fig:3Rhos} displays smoothed histograms of realizations of $I(\mathbf{y}^{\ast})$ from 2,000 replications each from the three respective posterior predictive distributions, $p(\mathbf{y}^{\ast} \mid \mathbf{y}) = \int_{\boldsymbol{\Theta}} f(\mathbf{y}^{\ast} \mid \boldsymbol{\theta})\pi(\boldsymbol{\theta} \mid \boldsymbol{y})d\boldsymbol{\theta}$, along with the approximate marginal posterior densities of $\rho$. For each value of $d$, the posterior of $\rho$ tends to concentrate near its upper bound and the posterior predictive densities of $I$ are nearly indistinguishable. Regardless of the value of $d$, $\rho$ adjusts accordingly and the overall spatial association is consistent with the data.
\begin{figure}[htb]  
    \centering
    \caption{Smoothed histograms of 2,000 realizations of Moran's $I$ from the corresponding posterior predictive distributions (left panel) and estimated marginal posterior distributions of $\rho$ under model (\ref{eqn:propModel}) (right panel). The dark vertical line in the left panel is at the observed value, $I(\mathbf{y})$.}
        \includegraphics[scale= 0.45, trim= 0in 18pt 0in 0in]{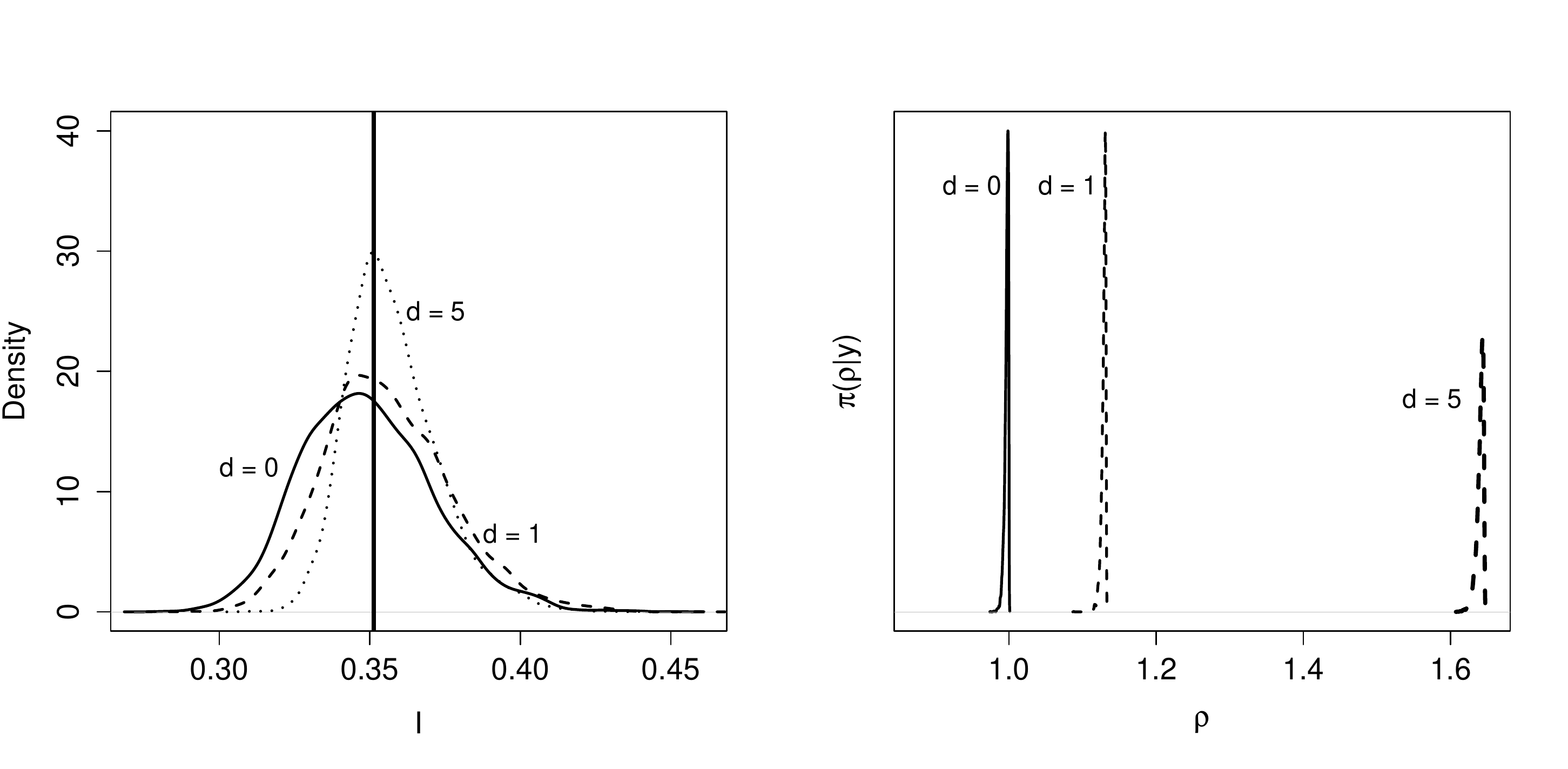}
        \label{fig:3Rhos}
\end{figure}

Generally speaking, including as many cases as possible when evaluating the posterior distribution is important. It is often the case that the dataset to be analyzed contains subsets of correlated observations among many independent observations. A standard CAR structure assumes every observation has at least one neighbor so that the only possibilities for dealing with both types of data are to force every data point into the model via an inappropriate neighborhood structure or to exclude isolated cases. Using an inappropriate dependence structure is clearly undesirable, and excluding isolated points can result in discarding a large proportion of the data. With our proposed approach, adjusting the weights with $d > 0$ in the denominator allows for the inclusion of all cases when evaluating the posterior distribution while simultaneously preserving the dependence among the cases sharing common networks as well as the independence of the isolated cases. In the sequel, we consider the performance of our model when $d = 0$ or $d = 1$. These are admittedly ad hoc values, and may not be appropriate for all situations.\\

\subsection{Computational Implementation}\label{sub:computation}

We facilitate the Gibbs sampling algorithm \citep*{GemanGeman84} by reparameterizing the model as $\tau^2 = \eta\sigma^2$. For ease of notation, define $\mathbf{D}^{\ast}_w := \mathbf{D}_w + d\mathbf{I}$. Then
\begin{eqnarray*}
\pi(\boldsymbol{\mu}, \boldsymbol{\gamma}, \sigma^2, \eta, \rho, p \mid \mathbf{y}) &\propto& (\sigma^2)^{-J-1}\exp\left(-\frac{(\mathbf{y}-\boldsymbol{\Gamma\mu})^T(\mathbf{y}-\boldsymbol{\Gamma\mu})}{2\sigma^2}\right)\\
    &&\times |\eta(\mathbf{D}^{\ast}_w-\rho\mathbf{W})^{-1}|^{-1/2}\exp\left(-\frac{\boldsymbol{\mu}^T(\mathbf{D}^{\ast}_w - \rho\mathbf{W})\boldsymbol{\mu}}{2\eta\sigma^2}\right)\\
    &&\times p^{J - \sum_{i=1}^J\gamma_i + \alpha-1}(1-p)^{\sum_{i=1}^J\gamma_i}(1+\eta)^{-2}I(\nu_1^{-1} < \rho < \nu_J^{-1}).
\end{eqnarray*}
\noindent The full conditional distribution of $\boldsymbol{\mu}$ is $N(\mathbf{R}\boldsymbol{\Gamma}\mathbf{y}, ~\sigma^2\mathbf{R})$, where $\mathbf{R} = (\boldsymbol{\Gamma} + \eta^{-1}(\mathbf{D}^{\ast}_w - \rho\mathbf{W}))^{-1}$ \citep*{CarlinLouis09}. To avoid matrix inversion with extremely large $J$, we update $\boldsymbol{\mu}$ element-wise. The full conditional distributions are given in the Supplementary Material.

We use rejection sampling to draw from the full conditional distribution of $\eta$. However, the importance ratio determining the acceptance probability is $\eta^2/(1+\eta)^2 \rightarrow 0$ as $\eta \rightarrow 0$. This means that the iterations can slow down on this step with candidate densities that concentrate on extremely small values of $\eta$. Possible alternatives are an adaptive Metropolis algorithm \citep*{HaarioEtAl05} or adaptive rejection sampling \citep*{GilksWild92}.

Since $\rho$ has bounded support, we follow the suggestion of \cite*{CarlinBanerjee03} and use slice sampling \citep*{Neal03} to sample from the full conditional distribution of $\rho$. Our experience is that the algorithm performs better with the ``doubling" procedure outlined by \cite*{Neal03} to adaptively determine good proposal intervals. Note that the determinant $|\mathbf{D}^{\ast}_w - \rho\mathbf{W}|^{1/2} \propto |\mathbf{I} - \rho(\mathbf{D}^{\ast}_w)^{-1/2}\mathbf{W}(\mathbf{D}^{\ast}_w)^{-1/2}|^{1/2}$ in the conditional density of $\rho$ can be quickly computed using the eigenvalues of $(\mathbf{D}^{\ast}_w)^{-1/2}\mathbf{W}(\mathbf{D}^{\ast}_w)^{-1/2}$, which do not depend on any unknown parameters. These eigenvalues will have already been computed if using the proposed model in (\ref{eqn:propModel}). Matrix computations can also be eased by calculating $\boldsymbol{\mu}^T\mathbf{D}^{\ast}_w\boldsymbol{\mu} = \sum_{j=1}^J(d + w_j.)\mu_j^2$ and $\boldsymbol{\mu}^T\mathbf{W}\boldsymbol{\mu}$ before searching for an acceptable update for $\rho$. These only need to be calculated once for each Gibbs iteration.

From (\ref{eqn:fullConditionals}) in the Supplementary Material, we can see the strong dependence between $\boldsymbol{\gamma}$ and $p$ in their full conditional distributions. On the $k^{\text{th}}$ iteration of a Gibbs sampling routine, if the sample draw $p^{(k)}$ is extremely close to 1, then most of the draws $\gamma_i^{(k)}, ~i= 1,\ldots,J$, will be zero. But then $\sum_i \gamma_i$ will be close to zero so that the conditional Beta density will concentrate close to 1, leading to another high value of $p$, and so on. Thus, in spite of the computationally convenient conditional conjugacy, an MCMC routine can get stuck in the region of the parameter space with most $\gamma_i = 0$, which slows. This situation can be ameliorated by reparameterizing the mixing proportion $p$ to eliminate boundary constraints, for instance a logit transformation $\pi := \log(p/(1-p))$, and using Langevin-Hastings proposals for $\pi$ to push the chain toward the posterior mode \citep[][Ch. 6]{GilksEtAl96}. Occasionally mixing in ordinary Metropolis proposals in place of Langevin-Hastings proposals offers further improvements when the chain is far from the mode \citep[][Ch. 3]{CarlinLouis09}.

The quantities of interest are the marginal posterior inclusion probabilities for each signal $j$, $P(\gamma_j = 1 \mid \mathbf{y})$. To estimate this from the posterior sample draws, we recognize that in Model (\ref{eqn:propModel}), $P(\gamma_j = 1 \mid \mathbf{y}) = E(p_j^{\ast})$, where the expectation is taken with respect to $\pi(\boldsymbol{\mu}, \boldsymbol{\gamma}_{(-j)}, \sigma^2, \eta, \rho, p, \mid \mathbf{y})$, and $p_j^{\ast} := P(\gamma_j = 1 \mid \boldsymbol{\mu}, \boldsymbol{\gamma}_{(-j)}, \sigma^2, \eta, \rho, p, \mathbf{y})$ is given by
\begin{equation}\label{eqn:CARModProbs}
p_j^{\ast} = \frac{(1-p)\varphi_{0,\sigma^2}(y_j - \mu_j)}{(1-p)\varphi_{0,\sigma^2}(y_j - \mu_j) + p\varphi_{0,\sigma^2}(y_j)}.
\end{equation}
This quantity can be estimated by $N^{-1}\sum_{i=1}^N\hat{p}_j^{\ast,(i)}$, where $\hat{p}_j^{\ast,(i)}$ is the plug-in estimate of $p_j^{\ast}$ evaluated with the $i^{th}$ draws $p^{(i)}, \mu_j^{(i)}, \sigma^{2,(i)}$ from the approximate joint posterior, and $N$ is the Monte Carlo sample size. Similarly, we can use (\ref{eqn:SBInclProbExp}) to estimate the inclusion probabilities under the Scott-Berger model using $p^{(i)}, \sigma^{2,(i)}, \tau^{2, (i)}$ drawn from the appropriate posterior. Both of these estimators are ``Rao-Blackwellized"  in the sense of \cite*{GelfandSmith90} and thus have smaller Monte Carlo variance than other more naive estimators \citep*[][Ch. 3]{CarlinLouis09}.

\section{Simulation Studies}\label{sec:Simulation}
To evaluate the performance of our model, we simulate a dataset in a manner similar to that which was carried out in \cite*{XiaoEtAl09}, resulting in similar correlation structure as sometimes arises between genes on chromosomes. We consider a situation in which 1,000 gene expressions are recorded for ten subjects. For the $j^{\text{th}}$ gene on the $i^{\text{th}}$ subject, $i= 1, \ldots, 10, j= 1, \ldots, 1000$, the observed expression level $X_{ij}$ is drawn from $N(\mu_{ij}, 1)$. Five of the subjects are taken as controls with baseline (i.e., null case) expression levels over all 1,000 genes, so that $\mu_{ij} = 0$ for $i=1, \ldots, 5, ~j= 1, \ldots, 1000$. The remaining five ``treatment" subjects' data are simulated so that 100 genes are differentially expressed. That is, for $i= 6, \ldots, 10$, $\mu_{ij} = 1.5, ~j= 1, \ldots, 20, 111, \ldots, 130, 211, \ldots, 230$ and $\mu_{ij} = -1.5, ~j= 311, \ldots, 330, 411, \ldots, 430$. For each control subject, we generate the gene expression levels by drawing the vector of observations $(X_{i,1}, \ldots, X_{i,1000})^T = \mathbf{X}_i \sim N(\mathbf{0}, \mathbf{I}), ~i= 1, \ldots, 5$. For the treatment group, the test statistics belonging to the null class are again simulated as i.i.d. standard normal. We model correlation among the differentially expressed (non-null) cases by drawing each group of twenty test statistics as $\mathbf{X}^{(k)}_i \sim N(\boldsymbol{\mu}^{(k)}, \boldsymbol{\Sigma}), ~i= 6, \ldots, 10$, where $\mathbf{X}^{(k)}_i, k= 1, \ldots, 5,$ is the $k^{th}$ cluster of non-null cases, (i.e., $\boldsymbol{\mu}^{(k)} = (1.5, 1.5, \ldots, 1.5)^T, k= 1, 2, 3,$ $\boldsymbol{\mu}^{(k)} = (-1.5, -1.5, \ldots, -1.5)^T, ~k= 4, 5$), and $\boldsymbol{\Sigma} = \{ 0.9^{|i-j|} \}_{i,j=1}^{20}$. This results in null cases that are independent of each other, but correlated non-null cases. Pooled $t$ statistics, $t_j, ~j = 1, \ldots, 1000$, are then calculated between the control and treatment conditions for each gene and subsequently normalized via probit transformation of the $p$-values \citep*{Efron10}, yielding test statistics $\mathbf{y} = (y_1, \ldots, y_{1000})^T$, where $y_j = \Phi^{-1}(F(t_j))$ with $F(\cdot)$ being the distribution function of the $t$ statistics.

We analyze the simulated data using both our proposed model and the Scott-Berger (SB) model assuming independence. In our model, we consider the sharing of information across genes using three different neighborhood structures. These neighborhoods are displayed graphically in Supplementary Figure \ref{fig:nbhds}. The associated adjacency matrices for these neighborhoods are
\begin{eqnarray*}
\mathbf{W}_1 &=& \left[ \begin{matrix}
            0 & 1 & 0 & \ldots & 0 & 0 \\
            1 & 0 & 1 & \ldots & 0 & 0 \\
            0 & 1 & 0 & 1 & \ldots & 0 \\
            \vdots & & & \ddots & & \vdots \\
            0 & 0 & 0 & \ldots & 1 & 0 \\
            \end{matrix} \right], ~~\mathbf{W}_2 = \left[ \begin{matrix}
            0 & 1 & 1 & 0 & 0 & \ldots & 0 & 0 & 0\\
            1 & 0 & 1 & 1 & 0 & \ldots & 0 & 0 & 0\\
            1 & 1 & 0 & 1 & 1 & \ldots & 0 & 0 & 0\\
            \vdots & & & & & \ddots & & \vdots &  \\
            0 & 0 & 0 & 0 & 0 & \ldots & 1 & 1 & 0 \\
            \end{matrix} \right],  \\
\mathbf{W}_3 &=& \left[ \begin{matrix}
            0 & 1 & \frac{1}{2} & \frac{1}{3} & 0 & 0 & \ldots & 0 & 0 & 0\\
            1 & 0 & 1 & \frac{1}{2} & \frac{1}{3} & 0 & \ldots & 0 & 0 & 0\\
            \frac{1}{2} & 1 & 0 & 1 & \frac{1}{2} & \frac{1}{3} & \ldots & 0 & 0 & 0\\
            \vdots & & & & & & \ddots & & \vdots \\
            0 & 0 & 0 & 0 & 0 & 0 & \ldots & \frac{1}{2} & 1 & 0 \\
            \end{matrix} \right]. \\
\end{eqnarray*}
For further comparison, we implement also the Significance Analysis for Microarrays (SAM) procedure as outlined in \cite*{Efron10}. SAM is a popular method for analyzing microarray data designed to approximately control the false discovery rate (FDR). For this procedure, we vary the FDR criterion between 0.05 and 0.15 to study performance across a range of FDR levels commonly used in practice. See \citet*[][Chapter 4]{Efron10} for further details of the SAM procedure.


Each neighborhood is one in which every location has at least one neighbor so that, in (\ref{eqn:propModel}), $w_j. > 0$ for all $j$. We take $d = 0$, reducing to the usual IAR model like the one considered in \citet*{BrownEtAl14}. To enforce sparsity {\em a priori}, we choose a high value of $\alpha$ in the prior on $p$ so that its prior probability mass is concentrated close to 1. We take $\alpha = 150$ here. In the SB model, we find that the best results are obtained using a uniform prior on $p$, with higher values of the shape parameter negatively affecting the model's ability to detect activation. Lastly, note that the data generating mechanism here is different from what is assumed under either our proposed model or the SB model, allowing us to study robustness, as well.

For both Bayesian models, we simulate the posterior distributions via MCMC. We implement the SB model using Gibbs sampling with nested rejection sampling steps for the non-standard distributions. To draw from the posterior of our proposed model, we use the conditional distributions given in the Supplementary Material for a Gibbs sampling procedure with nested rejection and slice sampling steps as described in Subsection 2.3. The algorithms are coded entirely in \verb|R| \citep*{RCore}. For our proposed model, a single chain uses a burn-in period of 5,000 iterations followed by an additional 10,000 sampling iterations, thinning to every fifth draw for a final sample of size 2,000. We run three chains in parallel and assess convergence with trace plots and scale reduction factors for selected parameters \citep*{GelmanRubin92}. Upon attaining approximate convergence, the retained draws from each of the three chains are combined for a final Monte Carlo sample size of 6,000. Similarly, we run parallel chains, assess convergence, and combine the draws to simulate the posterior distribution of the SB model. The posterior inclusion probabilities are estimated using (\ref{eqn:SBInclProbExp}) for the SB model and (\ref{eqn:CARModProbs}) for our proposed model. To select simulated genes as differentially expressed, we threshold the estimated posterior inclusion probabilities at 0.95.

Table \ref{tab:simArrayErrors} displays the empirical error rates for the SB model, the SAM procedure, and our model using $\mathbf{W}_1$, $\mathbf{W}_2$, and $\mathbf{W}_3$ above as neighborhoods. Of these models, we observe that the SB model results in the highest overall misclassification proportion, due to the false non-discoveries. In this case, the SB model is overcorrecting for multiplicity to the point that no cases are selected at all (hence the identically zero false discovery proportion). The SAM procedure performs slightly better in terms of non-discoveries and overall misclassification proportion, at the expense of false positives accounting for 13\% - 16\% of all discoveries. In contrast, our proposed model performs better, regardless of the selected neighborhood structure. The first-order neighborhood with unit weights ($\mathbf{W}_1$; top illustration in Supplementary Figure \ref{fig:nbhds}) performs best both in terms of non-discoveries and false discoveries, but all of the error rates are very close when compared to the other two approaches.

\begin{table}[h!]
    \centering
    \caption{False non-discovery proportions (FNP), false discovery proportions (FDP), and misclassification proportions (MCP) for the simulated gene expression data}
    \begin{tabular}{c c c c c c c c}
        & SB & SAM(0.05) & SAM(0.10) & SAM(0.15) & CAR ($\mathbf{W}_1$) & CAR ($\mathbf{W}_2)$ & CAR ($\mathbf{W}_3$) \\
        \hline
        FNP \hfill \vline & 0.1000 & 0.0938 & 0.0883 & 0.0856 & 0.0546 & 0.0577 &  0.0616 \\
        FDP \hfill \vline & 0.0000 & 0.1250 & 0.1333 & 0.1579 & 0.0000 & 0.0217 &  0.0238 \\
        MCP \hfill \vline & 0.1000 & 0.0940 & 0.0890 & 0.0870 & 0.0520 & 0.0560 &  0.0600 \\
        \hline
    \end{tabular}\label{tab:simArrayErrors}
\end{table}

Figure \ref{fig:ROCs} displays the empirical receiver operating characteristic (ROC) curves for the SB, SAM, and CAR($\mathbf{W}_1$) models. The ROC curves for the CAR($\mathbf{W}_2$) and CAR($\mathbf{W}_3$) models are virtually indistinguishable from the CAR($\mathbf{W}_1$) curve and so are not displayed. We see that the SB testing model and the SAM procedure are comparable in terms of overall discriminatory power. The approximate areas under the curves (AUC) are given in Table \ref{tab:AUCs}. While both of these procedures yield sensitivity slightly superior to that of our model at lower levels of specificity (less than about 0.40), ours still captures more area under its curve. In particular, note that our model attains a very sharp increase in sensitivity at very high levels of specificity.
\begin{figure}[h!]  
    \centering
    \caption{Empirical ROC curves for the simulated microarray data}
        \includegraphics[scale= 0.4, trim= 0in 18pt 0in 0in]{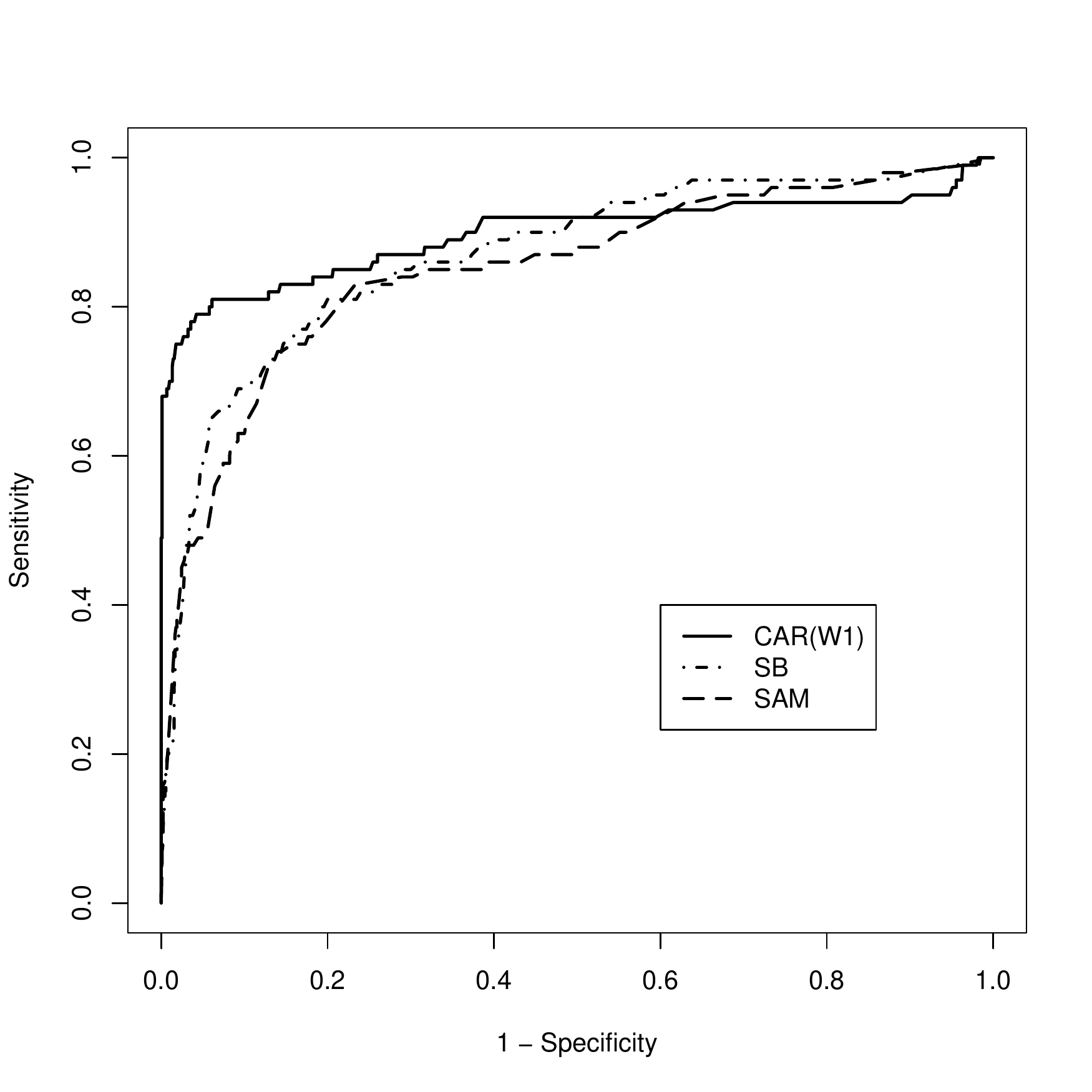}
        \label{fig:ROCs}
\end{figure}
\begin{table}[htb]
    \centering
    \caption{Approximate areas under each ROC curve displayed in Figure \ref{fig:ROCs}.}
    \begin{tabular}{c | c}
        Procedure & AUC \\
        \hline
        CAR($\mathbf{W}_1$) & 0.8969 \\
        SB & 0.8686 \\
        SAM & 0.8502 \\
    \end{tabular}\label{tab:AUCs}
\end{table}

Insight into the reasons for the discrepancies between our model and the SB model can be gained by examining the smoothed approximate posterior densities of $p$, $\sigma^2$, and $\tau^2 = \eta\sigma^2$, displayed in Figure \ref{fig:smthPost}. Incorporating the dependence results in much stronger Bayesian learning about these parameters. In this simulation, 100 out of 1,000 cases are non-null, so that we expect $p$ to be large, though not exactly 0.9 since correlation among the cases reduces the effective sample size \citep*[][Ch. 3]{CarlinLouis09}. That is indeed the case under our model. The SB model, on the other hand, results in considerably lower estimates of $p$, along with weakly identified distributions of $\sigma^2$ and $\tau^2$. This weak identifiability contributes to the error rates observed in Table \ref{tab:simArrayErrors}.
%
\begin{figure}[htb]  
    \centering
    \caption{Smoothed posterior estimates of $p$, $\sigma^2$, and $\tau^2$ for the simulated microarray data.}
        \includegraphics[scale= 0.35, trim= 0in 18pt 0in 0in]{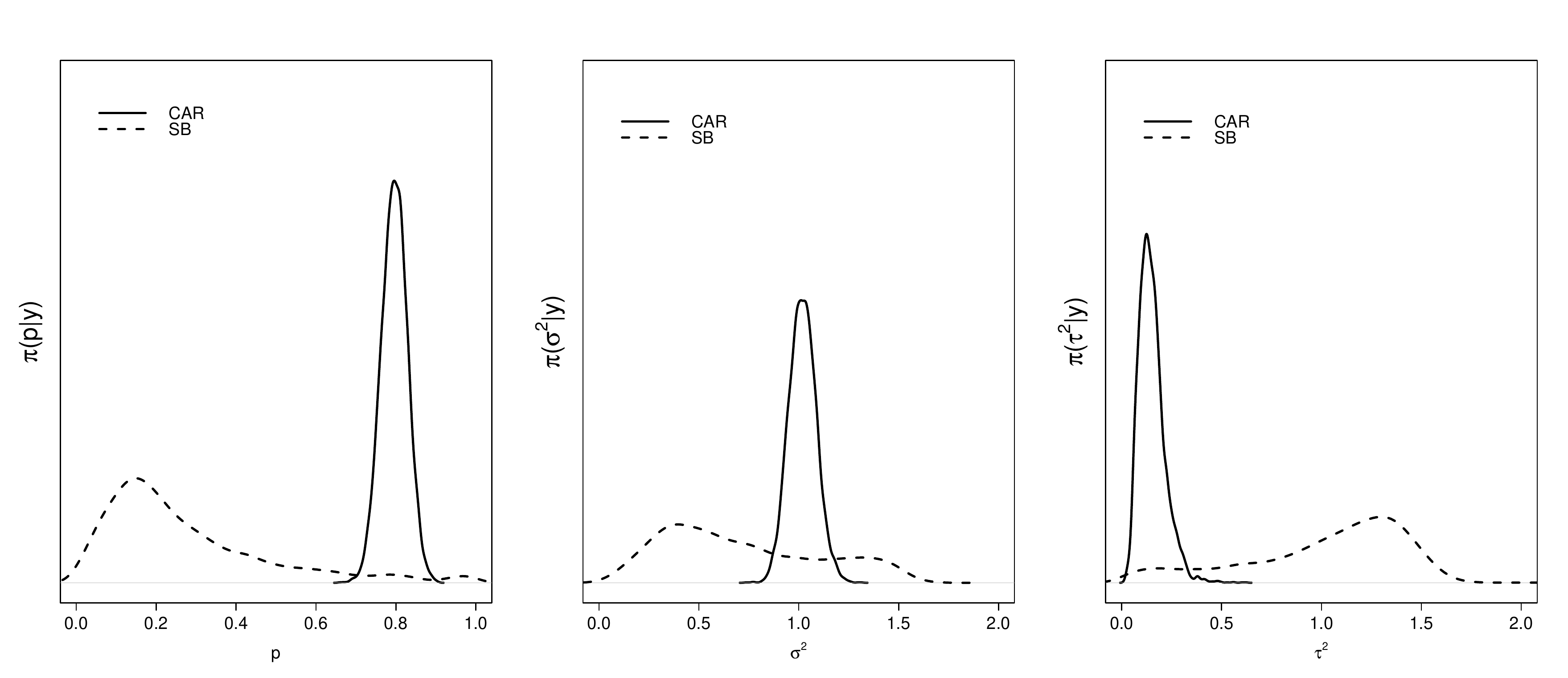}
        \label{fig:smthPost}
\end{figure}

Supplementary Figure \ref{fig:PiZi} plots the estimated inclusion probabilities versus the observed test statistics. All of the test statistics are assigned relatively high probabilities of being non-null under the SB model, but the lack of information about $\sigma^2$ and $\tau^2$ prevent distinctions being drawn that are strong enough to pass a 0.95 inclusion probability threshold. Our approach, on the other hand, results in stronger statements about the likelihoods of cases being non-null. The `jagged' quality of the curve corresponding to the CAR($\mathbf{W}_1$) model is due to the estimated inclusion probabilities being not a function of the $y_j$ values alone, but also of their location with respect to nearby observed values. To see this, consider the circled point in Supplementary Figure \ref{fig:PiZi}, which corresponds to the identified statistic in the graphical depiction of the test statistics in Supplementary Figure \ref{fig:pointZ}. This relatively extreme observation is in the middle of uninteresting test statistics. This is a truly null case, so the incorporation of a local dependence structure helps to prevent a false discovery.
%
%

In addition to neighborhoods determined by physical adjacency, it may be desirable to facilitate the sharing of information within gene sets such as those used in enrichment analysis \citep*{SubramanianEtAl05}. To investigate this scenario, we simulate an additional microarray dataset with expression characteristics similar to the simulation carried out in \cite*{EfronTib07}. We again consider a collection of 1,000 genes, with genes 11-20, 111-130, 211-230, 311-330, 411-430 defining five different gene sets. The set consisting of genes 111-130 is simulated as differentially expressed by drawing them independently from $N(2.5, 1)$; similarly, the genes 411-430 are drawn independently from $N(-1.5, 1)$. The remaining genes, including those in the remaining gene sets, are considered null cases and all drawn from $N(0,1)$. To distinguish dependence determined by pathway membership from dependence determined by physical adjacency, the genes are labeled and randomly permuted so that genes sharing common pathways are not adjacent in the physical sense. Note that many genes are members of no pathway at all and thus are isolated.

Suppose a researcher were to mistakenly assume the dependence structure for these data followed the same pattern as the previous example in which every gene is correlated with its physical neighbors, meaning the usual proper IAR with adjacency matrix $\mathbf{W}_1$ can be used with $d = 0$ in (\ref{eqn:propModel}) as before. We compare the performance of this CAR structure to our proposed approach with neighborhoods determined by pathway membership. In the latter case, the adjacency matrix $\mathbf{W}$ is determined by defining genes that are in the same set to be neighbors. To include all observations without reducing the rank of the precision matrix, we set $d = 1$ in (\ref{eqn:propModel}) so that the marginal distributions of isolated $\mu_j$ become $N(0, \tau^2)$. We implement both models using MCMC with the same burn-in and sampling settings as the previous simulation.

Supplementary Figure \ref{fig:PathVsAdjROC} displays the empirical ROC curves from our model under both neighborhood assumptions, in which we can see that making incorrect assumptions about the neighborhood structure severely inhibits the model's discriminatory power. In fact, thresholding the posterior inclusion probabilities at 0.95 as before results in no cases being identified at all under the physical adjacency assumption. The misspecification of the correlation results in the model overestimating the noise variability in the data, as is clear upon examination of the histogram in Supplementary Figure \ref{fig:PathVsAdjHist}. Superimposed on the histogram are two mean-zero normal densities with variances equal to the posterior means of $\sigma^2$ under both models. The overestimation of the spread of the null distribution results in poor identification of the non-null cases, indicated with the dark tick marks along the $x$-axis. Incorporating a more appropriate neighborhood structure results in improved estimation of the variance components and thus superior discrimination among cases.

Even with knowledge of an approximately correct dependence structure among pathways, a temptation might be to use a standard CAR model by discarding the isolated cases. While this is obviously better than using an entirely incorrect neighborhood assumption, the isolated points still provide information about the parameters common to both the null and non-null components of the data distribution and hence potentially useful information can be needlessly discarded. Consider the smoothed marginal posterior densities of $p$, $\sigma^2$, and $\tau^2$ resulting from both approaches, displayed in Figure \ref{fig:simPathFVsRDens}. Eliminating the isolated cases reduces $J$ in (\ref{eqn:propModel}), so we obtain considerably less posterior concentration about the error variance, which in turn affects the amount of information available to estimate the second variance component, $\tau^2$. As is the case in any hypothesis testing scenario, the operating characteristics are directly affected by the amount of information we have about the error variability. This results in the ``borderline" cases being misclassified as noise at a 0.95 inclusion probability threshold, thus increasing the false non-discovery proportion. (See Supplementary Figure \ref{fig:simPathFVsRProbs}.) The false discovery, false non-discovery, and misclassification proportions at the 0.95 threshold with and without isolated cases are given in Table \ref{tab:FullVsRedErrors}.
\begin{figure}[htb]  
    \centering
    \caption{Smoothed posterior estimates of $p$, $\sigma^2$, and $\tau^2$ with and without the isolated cases.}
        \includegraphics[scale= 0.4, trim= 0in 18pt 0in 0in]{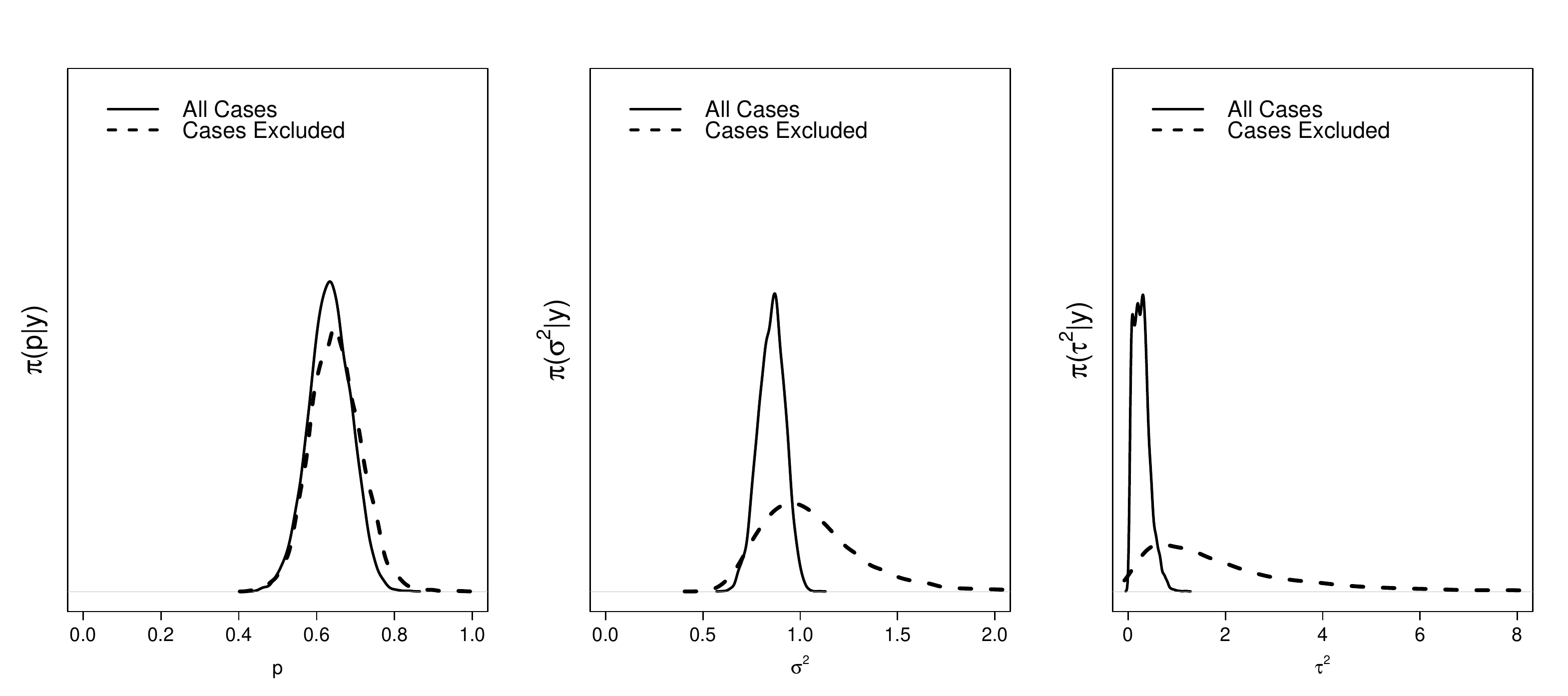}
        \label{fig:simPathFVsRDens}
\end{figure}
\begin{table}[h!]
    \centering
    \caption{Error rates for the simulated pathway example using pathway-based neighborhoods.}
    \begin{tabular}{c c c}
        & With Isolated Cases & Without Isolated Cases \\
        \hline
        FNP \hfill \vline & 0.0083 & 0.1781\\
        FDP \hfill \vline & 0.0000 & 0.0000 \\
        MCP \hfill \vline & 0.0080 & 0.1300 \\
        \hline
    \end{tabular}\label{tab:FullVsRedErrors}
\end{table}

In addition to detection, there is often an interest in estimating the true signal strengths of the non-null cases. Figure \ref{fig:muEsts} plots the smoothed approximate posterior densities and approximate 95\% credible intervals about the signals, $\mu$, for two typical non-null cases in the simulated gene pathway data. Again, the posteriors are evaluated with and without the isolated cases using neighborhoods defined by pathway membership. The true signal strengths (i.e., the means of the distributions of the test statistics) for these two cases are $E(Z) \approx -2.78$ (top panel) and $E(Z) \approx 2.12$ (bottom panel), indicated in the plots by vertical lines. The Figure illustrates the reduction in posterior uncertainty that can be obtained by including all of the data points. For the top panel, the approximate 95\% credible intervals with and without the isolated cases are $[-3.36, -2.40]$ and $[-3.77, -2.12]$, respectively. For the bottom panel, the intervals are $[1.70, 2.66]$ and $[1.24, 3.07]$ with and without the isolated cases, respectively. Table \ref{tab:CIWidths} displays the average widths of the approximate 95\% credible intervals over all of the cases in both non-null pathways with and without the isolated observations. By including the isolated points with an appropriate neighborhood structure afforded by our proposed model, we attain an approximate four-fold increase in the precisions of the signal estimates.
\begin{figure}[htb]  
    \centering
    \caption{Smoothed approximate posterior densities of the signal strengths $\mu_j$ for two non-null cases in the simulated gene pathway example. The bars at the top are the approximate 95\% posterior credible intervals. The vertical lines indicate the true means of the non-null distributions of the test statistics.}
        \includegraphics[scale= 0.6, trim= 0in 18pt 0in 0in]{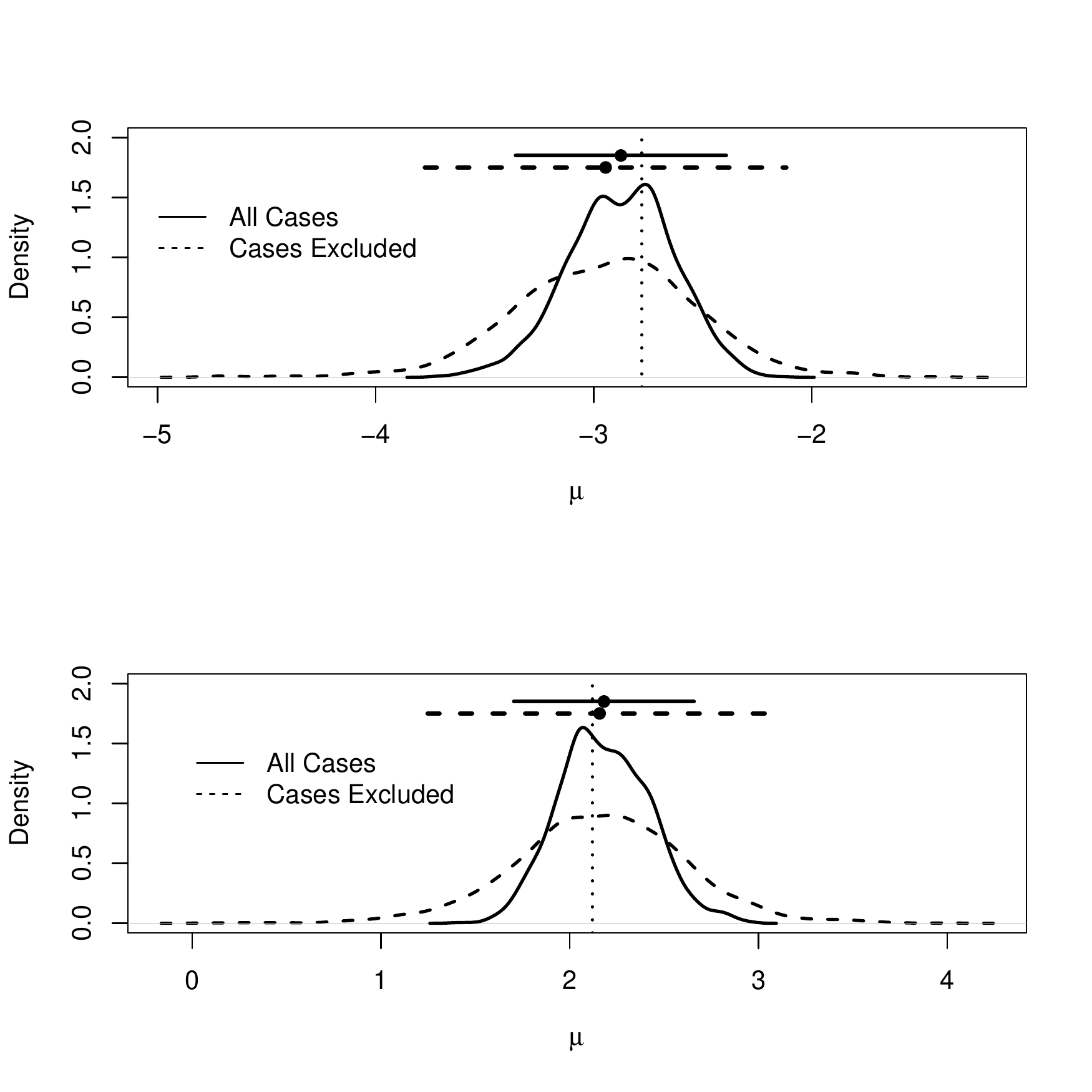}
        \label{fig:muEsts}
\end{figure}
\begin{table}[htb]
    \centering
    \caption{Average widths of the approximate 95\% credible intervals for the non-null pathways in the simulated gene pathway example.}
        \begin{tabular}{c | c c}
            Pathway & All Cases & Cases Excluded \\
            \hline
            2 & 0.9603 (0.0137) & 1.708 (0.0604) \\
            5 & 0.9338 (0.0105) & 1.700 (0.0433) \\
        \end{tabular}\label{tab:CIWidths}
\end{table}

In practice, the most appropriate neighborhood structure to use in the CAR model may not be known. For gene microarray, though, only biologically meaningful dependence structures would usually be considered so that one would not be faced with completely unguided choices. Even among interpretable neighborhoods, we still might be forced to compare competing neighborhood assumptions in an effort to determine which is best. Consider again the simulated pathway data, only we do not know whether the dependence is among biologically-determined pathways, or if it is a function of physical adjacency. In this case, we can gauge the spatial dependence by considering Moran's $I$ statistic under different neighborhood assumptions, whence the competing models can be examined by looking at the strength of estimated spatial association according to each. For the simulated pathway data, we have $I(\mathbf{y}) = 0.0041$ for the (incorrect) adjacency assumption, and $I(\mathbf{y}) = 2.0747$ for the (correct) pathway assumption. The lack of association under the adjacency structure is indicative of the invalid assumption, since we would suspect there to be some kind of association (otherwise there's no need for a CAR structure at all). Further, we can investigate competing models' predictive capabilities through the use of root mean square predictive error (RMSPE) and the Wantanabe-Akaike information criterion \citep*[WAIC;][]{Wantanabe10}, the latter of which is asymptotically equivalent to Bayesian leave-one-out cross validation but averages over the posterior distribution instead of relying on point estimates, unlike AIC or DIC \citep*{GelmanEtAl14}. Table \ref{tab:PathWAIC} displays RMSPE and WAIC for both models, where we see that the correct neighborhood structure is favored.
\begin{table}[htb]
    \centering
    \caption{Root mean square predictive error (RMSPE) and Wantanabe-Akaike information criterion (WAIC) for both neighborhood models in the simulated pathway example.}
    \begin{tabular}{c c c}
        & Pathways & Adjacency \\
        \hline
        RMSPE \hfill \vline & 1.498 & 1.578 \\
        WAIC \hfill \vline & 2794.614 & 3016.502
    \end{tabular}\label{tab:PathWAIC}
\end{table}

The true correlation structure among gene expression data is complex. The implementation of our proposed model (and similar CAR models), however, requires neighborhood structures to be specified {\em a priori}, and this specification can potentially be incomplete or simplistic. Therefore, to study the performance of our proposed model under partially incorrect neighborhood assumptions, we simulate 1,000 expression levels over ten subjects (five treatment, five control) with genes 11-20, 111-130, 211-230, 311-330, 411-430 defining five different gene sets as before. We again suppose that the second and fifth sets are enriched in the treatment group. In contrast to the previous simulation, assume for each treatment subject that the correlation among signals is induced according to a directed graph, as might occur in metabolic pathways in which a signal originates from a single gene and cascades via several networks to other genes downstream \citep*{StingoEtAl11}. A ``parent" gene is selected at random from each of the two active gene sets, and the signals for these parent genes are simulated as $\mu_p \sim N(2.5, 0.75^2)$. Within each pathway, seven additional child nodes are selected at random and their signals are taken to be $\mu_{c_1}, \ldots, \mu_{c_7} \stackrel{\text{iid}}{\sim} N(0.92\mu_p, 0.75^2)$. The remaining signals are drawn independently from $N(0.92\sum_{s=1}^7\mu_{c_s}, 0.75^2)$. Within each gene set, the observed expression levels are taken to be $X_{ik} \stackrel{\text{indep.}}{\sim} N(\mu_k, 1), ~i= 1,\ldots,10$, where $k$ indexes over all genes, including the parent nodes and child nodes, and $i$ indexes the subjects. To represent incompleteness, suppose the neighborhood structure determining $\mathbf{W}$ omits two genes that are actually members of one of the active sets, and likewise for one of the inactive sets. Lastly, we randomly select 30 isolated genes and draw their expression levels from $N_{30}(\mathbf{0}, \boldsymbol{\Sigma})$ for all subjects, treatment and control, where $\boldsymbol{\Sigma} = \{0.9^{|i-j|}\}_{i,j=1}^{30}$. Our proposed model thus incorrectly assumes uniform correlation within each pathway, uses incomplete pathway definitions, and ignores correlation among a subset of genes with no known pathway membership.

Table \ref{tab:cplxArrayErrors} displays the empirical false nondiscovery proportions, false discovery proportions, and overall misclassification proportions for these simulated data under the Bayesian independence testing model and our proposed pathway-determined CAR model with 0.95 posterior probability threshold, as well the SAM procedure with three common thresholds. Despite the model misspecification, our model performs comparably to the Bayesian approach assuming (incorrectly) independence and the generally applicable SAM procedure. In fact, our approach still yields the smallest empirical misclassification proportion, though they are all very close. Even though we have partially incorrect assumptions about the correlation structure, our approach is still superior to that of assuming complete independence in terms of predictive capability, as evident in the lower RMSPE and WAIC values, also displayed in Table \ref{tab:cplxArrayErrors}. Further, the assumed neighborhood structure in our model predicts non-zero spatial correlations that are fairly consistent with those observed in the actual data under the same structure. This is evident in Supplementary Figure \ref{fig:simCplxPPDMoranI}, which plots a histogram of realized $I(\mathbf{y})$ values from the posterior predictive distribution along with the value observed in the actual data. Under the assumed neighborhood structure, $P(I(\mathbf{y}^{\ast}) \geq I(\mathbf{y}) \mid \mathbf{y}) \approx 0.1175$.
\begin{table}[htb]
    \centering
    \caption{Error rates and predictive capabilities of the independence Bayesian testing model (SB) and the CAR testing model under incorrect correlation assumptions, along with error rates from the SAM procedure with three common thresholds. SAM is not used for prediction, so the RMSPE and WAIC are not applicable.}
    \begin{tabular}{c c c c c c c c}
        & SB & CAR & SAM(0.05) & SAM(0.10) & SAM(0.15)\\
        \hline
        FNP \hfill \vline & 0.0123  & 0.0093 & 0.0093 & 0.0093 & 0.0093\\
        FDP \hfill \vline & 0.0000  & 0.0313 & 0.0606 & 0.0882 & 0.1143\\
        MCP \hfill \vline & 0.0120  & 0.0100 & 0.0110 & 0.0120 & 0.0130\\
        \hline
        RMSPE \hfill \vline & 1.3256 & 0.7590 & $-$ & $-$ & $-$\\
        WAIC \hfill \vline & 3470.8036 & 2778.6636 & $-$ & $-$ & $-$ \\
    \end{tabular}\label{tab:cplxArrayErrors}
\end{table}

These simulation results indicate that overall performance of the mixture prior can be improved by using common information across local neighborhoods, when it is available. This improvement is due to the induced penalty on the inclusion probabilities of statistics surrounded by uninteresting observations, and to improved estimation of the mixing proportion and variance components in the data. By choosing the neighborhood weights appropriately, we demonstrate how our model can accommodate isolated genes that have no neighbors. Our proposed approach is evidently superior to an approach using a conventional CAR model, even when the correct pathway information is used to define the neighborhoods but isolated points are discarded. Incorporating spatial dependence and isolated cases results in much sharper Bayesian learning in the posterior distribution, improving inference and reducing uncertainty. Simple diagnostics such as Moran's $I$ under different assumed neighborhood structures can be helpful when competing neighborhood structures are available, as well as considering predictive capabilities through measures such as WAIC, which approximates out-of-sample predictive error while averaging over the posterior distribution. We illustrate also that even simplistic correlation assumptions in the Bayesian testing approach still perform competitively with alternatives such as the SAM procedure while predicting dependence features that are consistent with the observed data. Judging from model fit criteria, partially incorrect correlation assumptions are better than ignoring them altogether in an independence model.\\

\section{Application}\label{sec:Application}
\subsection{E. Coli Data}
As one application of our proposed model, we consider the microarray expression data from \cite*{XiaoEtAl09}. This dataset contains transcriptional activity on the {\em Escherichia coli} chromosome measured as log ratios of transcript abundances between a control and various chemical, physiological, and genetic perturbations comprising 53 experimental conditions. The observed gene expression levels are the average log ratios across conditions. Here, we have 4 replicate measurements on 4,276 genes. Test statistics are calculated by dividing the mean difference by the standard error plus a small constant to reduce the effect of extreme observations, as is done in the SAM procedure. These statistics are probit transformed to yield equivalent $z$ statistics.

The {\em E. coli} chromosome has been shown to have correlated expression levels according to gene location, but with a circular structure so that the first and last genes on the chromosome are considered neighbors \citep*{XiaoEtAl09}. This structure leads us to use the adjacency matrix obtained by replacing the last elements of the first row and first column of $\mathbf{W}_1$ in Section 3 by 1. We take $d = 0$ since each point has two neighbors.

We simulate the posterior distributions of the independence model and our CAR model using the MCMC algorithms described in Sections 2 and 3. The resulting posterior activation probabilities are thresholded at 0.99 to select genes as being differentially expressed. To evaluate performance, we compare the genes selected as differentially expressed to a list of 41 genes identified in \cite*{Macnab92} as having a known or suspected function in the {\em E. coli} chromosome. This list serves as a reference with which we calculate approximate false discovery and false non-discovery proportions.

To illustrate the effect of the shape parameter in the prior for $p$ in our model, we repeatedly simulate the posterior distribution while varying $\alpha$ over selected values between 1 and 1775. For the independence model, we again take $p \sim \text{Unif}(0,1)$. Table \ref{tab:eColiErrorRates} gives the consequent empirical error rates. For lower values of $\alpha$, the sensitivity results in generally higher error rates compared to that from assuming independence. However, for higher $\alpha$, we are able to attain uniformly better performance, with all three error rates outperforming the independence model. The effect of $\alpha$ on the marginal posterior distributions of $p$ and $\rho$ can also be seen in Supplementary Figure \ref{fig:PostDensAlpha}. Higher $\alpha$ values result in higher estimated values of $p$, as expected, but they also result in sharper posterior inferences about both $p$ and $\rho$. We remark that the false discovery proportions are all quite high here. As this analysis is based on a real dataset, though, there is no way of knowing which of these are true false discoveries. Indeed, the high FDP could be a consequence of \cite*{Macnab92} listing the {\em most} interesting genes, not necessarily {\em all} interesting genes.\\
\begin{table}[h!]
    \centering
    \caption{Error rates for the E. Coli data \citep{XiaoEtAl09} under the independence model and the CAR testing model under selected values of $\alpha$ in $\pi_p(p)$. For the CAR model, $d = 0$. The independence model uses a uniform prior on $p$.}
    \begin{tabular}{l c c c c c}
        & Independence & $\alpha = 1$ & $\alpha = 500$ & $\alpha = 1000$ & $\alpha = 1775$ \\
        \hline
        FNP \hfill \vline & 0.0021 & 0.0005 & 0.0010 & 0.0009 & 0.0012 \\
        FDP \hfill \vline & 0.4074 & 0.7821 & 0.5316 & 0.4219 & 0.3793 \\
        MCP \hfill \vline & 0.0072 & 0.0332 & 0.0108 & 0.0072 & 0.0063 \\
        \hline
    \end{tabular}\label{tab:eColiErrorRates}
\end{table}

\subsection{Male vs. Female Lymphoblastoid Cell Data}
For an additional application under a different neighborhood structure, we consider mRNA expression profile data collected from lymphoblastoid cell lines derived from fifteen males and seventeen females. This dataset was analyzed in \cite*{SubramanianEtAl05} with gene set enrichment analysis (GSEA), who sought to identify gene sets enriched in males (male $>$ female) and enriched in females (female $>$ male). Each cell line contains measurements on 22,283 genes. The existing catalog for this dataset includes 319 cytogenetic gene sets, 24 for each of the 24 human chromosomes, and 295 associated with cytogenetic bands. For our analysis, we again calculate two-sample $t$-statistics and normalize to obtain $z$-statistics.

We consider two variants of our proposed CAR testing approach, both of which use the 319 gene sets to define the neighborhoods, but with one model including the isolated points and the other excluding isolated points so that a conventional CAR model is applicable. With the isolated cases included, we set $d = 1$ in (\ref{eqn:propModel}); $d = 0$ when there are no isolated cases. Otherwise, both models are the same. We take $\alpha = 1$ in the prior on $p$ so that it becomes uniform. The MCMC algorithm, identical under both models, uses a burn in of 25,000 iterations, followed by 10,000 sampling iterations, of which every fifth draw is retained. The posterior inclusion probabilities are calculated using (\ref{eqn:CARModProbs}) and thresholded at 0.99 to identify differentially expressed genes.

The estimated posterior inclusion probabilities under both models are quite similar, though not exactly the same. This can be seen in Supplementary Figure \ref{fig:gseaGeneProbs}, which plots the posterior inclusion probabilities versus the test statistics for both models. The differences result in a couple of disagreements on the selection of differentially expressed genes, listed in Table \ref{tab:GSEACasesESS}.
\begin{table}[htb]
    \centering
    \caption{List of genes selected under the generalized CAR model including isolated cases (gCAR) or a conventional CAR model with isolated cases excluded (CAR). The $\times$ symbol indicates selection under the gCAR with isolated points, the $\bullet$ symbol indicates selection when ignoring isolated points. Also listed are the effective sample sizes (ESS) of $\mu_j^{(1)}, \ldots, \mu_j^{(2000)}$ for each gene from the MCMC output.}
    \begin{tabular}{l | c c | c c c c c c}
        & \multicolumn{2}{c}{ESS} \vline & \multicolumn{6}{c}{Pathway} \\
        Gene & gCAR & CAR & \verb|chrX| & \verb|chrXq13| & \verb|chrXp22| & \verb|chrY| & \verb|chrYp11| & \verb|chrYq11|\\
        \hline
        \verb|201028_s_at|  & 1855.221 & 1164.651  & $\bullet$ & & $\bullet$ & $\bullet$ & $\bullet$ & \\
        \verb|201909_at|  & 2000.000 & 1012.763 & & & & $\times$$\bullet$ & $\times$$\bullet$ &  \\
        \verb|204409_s_at|  & 2000.000 & 557.216 & & & & $\times$$\bullet$ & & $\times$$\bullet$ \\
        \verb|204410_at|  & 2106.911 & 1221.072 & & & & $\times$$\bullet$ & & $\times$$\bullet$  \\
        \verb|205000_at| & 2000.000 & 694.079 & & & & $\times$$\bullet$ & & $\times$$\bullet$  \\
        \verb|205001_s_at| & 2000.000 & 2000.000 & & & & $\times$$\bullet$ & & $\times$$\bullet$ \\
        \verb|206624_at| & 2000.000 & 1561.820 & & & & $\times$$\bullet$ & & $\times$$\bullet$  \\
        \verb|206700_s_at| & 2148.696 & 615.271 & & & & $\times$$\bullet$ & & $\times$$\bullet$ \\
        \verb|214131_at| & 1594.909 & 1873.845 & & & & $\bullet$ & & $\bullet$ \\
        \verb|214218_s_at| & 2000.000 & 85.484 & $\times$$\bullet$ & $\times$$\bullet$ & & & &  \\
        \verb|214983_at| & 1778.600 & 2000.000 & & & & $\bullet$ & &  \\
        \verb|221728_x_at| & 1936.754 & 131.264 & $\times$$\bullet$ & $\times$$\bullet$ & & & &  \\
        \verb|203974_at|  & 2000.000 & & \multicolumn{6}{c}{ $\times$ (Isolated case)}\\
    \end{tabular}\label{tab:GSEACasesESS}
\end{table}

The disagreements between the two approaches can partly be explained by the reliability of the estimates themselves. Table \ref{tab:GSEACasesESS} lists the effective sample sizes (ESS) of the MCMC draws of the signals for each identified gene, $\mu_j^{(1)}, \ldots, \mu_j^{(2000)}$ \citep*{KassEtAl98}. The ESS is an approximation of the number of {\em independent} pieces of information about a parameter produced by an MCMC algorithm, where lower numbers reflect higher autocorrelation in the chain and hence slower convergence. The differences between the two approaches considered here are substantial. With a couple of exceptions, the retained values from the Markov chain under the conventional CAR assumption are more highly correlated than in the gCAR, thus reducing the amount of available information about these parameters from a fixed number of iterations. On the other hand, including the isolated cases in this instance generally results in the Markov chain samples being almost as good as an i.i.d. sample from the target distribution.

Table \ref{tab:GSEACasesESS} also indicates the pathway membership for each case. There are six gene sets in which the identified individual genes appear, and the clustering of genes along the X and Y chromosome is apparent. If one were to use the approach of simply declaring as enriched the pathways including the individual differentially expressed genes, our results would agree closely with those found in the GSEA. In particular, the GSEA identified the Y chromosome (\verb|chrY|) and two Y bands with at least 15 genes (\verb|chrYp11, chrYq11|) as being associated with higher expression levels in males. Two of the genes selected by both models appear in the \verb|chrX| and \verb|chrXq13| gene sets. These genes are associated with the set of X chromosome inactivation genes, which is expected to be enriched in females. Note that our proposed gCAR enabled us to identify an isolated gene (\verb|203974_at|) that does not appear in any of the predefined gene sets. This gene would have been missed entirely if we were to use a conventional CAR structure inside of our Bayesian testing model, since it would have been discarded from the analysis. We notice also a particular gene selected only by the conventional CAR that appears in both the X and Y chromosome. This curious case could be a false positive, though, as the slower MCMC convergence under the conventional CAR model makes the posterior inference less reliable than its gCAR counterpart.

The applications presented here illustrate two different approaches to defining neighborhoods across which information may be shared when searching for non-null cases under our Bayesian testing model. While the results are sensitive to the choice of hyperparameters as well as the threshold, it is apparent that ``good" choices can lead to desirable operating characteristics. Applying our proposed approach to these data yields results consistent with past analyses. These results demonstrate our model's ability to harness shared information between cases without sacrificing the possibility of identifying independent cases, something that would not be possible under the conventional CAR assumption. In this analysis, we found that including the isolated cases substantially reduced the autocorrelation. This results in quicker convergence and much greater sampling efficiency than is obtained from a conventional CAR model. This is an especially important consideration when performing MCMC in a large-scale setting, where the computational burden limits the feasible number of iterations. Our results suggest that there could be a positive effect on the sampling efficiency of an MCMC algorithm by including isolated, independent cases in our generalized CAR structure. This is possibly an interesting topic for future research that we do not pursue here.\\

\section{Discussion}\label{sec:Conclusions}
We present a unified approach to correlated Bayesian testing whereby isolated cases and neighboring cases can be included in the same analysis. This allows for improved estimation of the signal strengths, the possibility of identifying isolated cases, and sharper posterior inferences about parameters of interest. We suggest some simple diagnostics that can aid a researcher in determining the most appropriate neighborhood structure when choosing among plausible models. When very little prior information is available concerning correlation structure, we note that there exist in the literature proposed techniques for discovering structure in high dimensional data such as sparse factor modeling \citep*{West03} and independent components analysis \citep*{Comon94}. We demonstrate the robustness of our approach to model misspecification by applying it to simulated data with a complex correlation structure in which the assumptions are partially incorrect. It performs competitively with well-established procedures.

The results presented here are seen to be quite sensitive to the choice of the shape parameter in the prior for the mixing proportion, $p$, in our proposed model. This is in part because large $\alpha$ values result in both prior and posterior concentration of $p$ about large values. These large values mean that the Gibbs sampler tends to visit sparser models more often, and thus parameters that only appear in non-null cases are updated less frequently. Certain applications necessitate the use of a prior that enforces known sparsity \citep*[e.g.,][]{West03, CarvalhoEtAl08}. While the choice of shape hyperparameter does have a considerable effect on subsequent inferences, we demonstrate how finding a good value leads to desirable operating characteristics. In working with our model, we found that an acceptable value of the shape parameter seems to depend upon the strength of the correlation across neighboring observations. The best way to choose this value or otherwise tune the prior to approximately match the true {\em a priori} non-null probability in the data is still an open problem worthy of further investigation.

A related point is the thresholding of the location-specific {\em a posteriori} inclusion probabilities. We use throughout this paper an informal 0.95 decision rule for selecting non-null cases. Decision rules in the Bayesian testing paradigm have been proposed through average risk optimization and consideration of the so-called Bayes false discovery rate (bFDR) \citep*[e.g.,][]{EfronTib02, TadesseEtAl05, BogdanEtAl08}. However, most results concerning the relationship between bFDR and frequentist error measures are based on assumed independence in the data, which we are not considering. Performance also is determined in part by specification of the prior probabilities of the hypotheses. However, the approach of \cite{ScottBerger06} on which our model is based enjoys the virtue of inducing an automatic multiplicity adjustment, even in the correlated case \citep*{BrownEtAl14}. While expression (\ref{eqn:CARModProbs}) allows our calculations to viewed as a fully Bayesian treatment of local false discovery rates \citep*[][Ch. 5]{Efron10}, much work remains to be done on establishing optimal decision rules for controlling different error rates under dependence.

This paper provides a glimpse at the possibility of facilitating more reliable inference by capturing (or at least approximating) the true dependence structures that are inherent in modern high-dimensional data. While this issue has garnered more interest in the recent literature \citep*[e.g.,][]{SmithFahr07, LiZhang10, StingoEtAl11, LeeEtAl14, ZhangEtAl14, ZhaoEtAl14} relatively limited work has been done on modeling nontrivial dependence in Bayesian models for signal detection, particularly in the high-dimensional setting where classical multiple testing approaches are no longer appropriate. Here we propose using Markov-type dependence structures in the continuous component of the spike-and-slab mixture with a Gaussian CAR model. An avenue of future work could be the exploration of other dependence structures. Work has been done on modeling complex measures of distance and covariance structures \cite*[e.g.,][]{DrydenKoloZhou09}, but there is a need for much further research toward building a flexible class of multiple testing models capable of dealing with a wide variety of dependence types.

\begin{appendices}
\section*{Appendix: Proof of Posterior Propriety of the Proposed Model}
    Let $\boldsymbol{\Theta} = (\boldsymbol{\mu}^T, \sigma^2, \tau^2, \rho)^T$ denote the parameter vector and let $F(\mathbf{y},\boldsymbol{\Theta},\boldsymbol{\gamma},p)$ be the joint distribution of the data and the parameters. Then for Model (\ref{eqn:propModel}), it suffices to show that
    \begin{equation}\label{eqn:jointInt}
    \int_{(\boldsymbol{\Theta},\boldsymbol{\gamma},p)}dF(\mathbf{y},\boldsymbol{\Theta},\boldsymbol{\gamma},p) = \sum_{\boldsymbol{\gamma} \in \{0,1\}^J} \int_p \int_{\boldsymbol{\Theta}} f(\mathbf{y},\boldsymbol{\Theta},\boldsymbol{\gamma},p)d\boldsymbol{\Theta} dp < \infty, ~\forall \mathbf{y} ~(a.e.).
    \end{equation}
    First, note that for any $\boldsymbol{\gamma} \in \{0,1\}^J$,
    \begin{eqnarray*}
    f(\mathbf{y},\boldsymbol{\Theta},\boldsymbol{\gamma},p) &=& (2\pi\sigma^2)^{-J/2}\exp\left(-\frac{1}{2\sigma^2}\sum_{j=1}^J(y_j-\gamma_j\mu_j)^2\right)\\
     && \times(2\pi)^{-J/2}|\tau^2 (\mathbf{D}^{\ast}_w-\rho\mathbf{W})^{-1}|^{-1/2}\exp\left(-\frac{1}{2\tau^2}\boldsymbol{\mu}^T(\mathbf{D}^{\ast}_w-\rho\mathbf{W})\boldsymbol{\mu}\right)\\
     && \times(\sigma^2+\tau^2)^{-2}\pi_{\rho}(\rho)\alpha p^{J-\sum_j\gamma_j+\alpha-1}(1-p)^{\sum_j\gamma_j}\\
     && \hfill \\
     &\equiv& f(\mathbf{y}, \boldsymbol{\Theta} \mid \boldsymbol{\gamma})\pi_{(\boldsymbol{\gamma},p)}(\boldsymbol{\gamma},p),
    \end{eqnarray*}
    where $\pi_{(\boldsymbol{\gamma},p)}(\boldsymbol{\gamma},p) = \alpha p^{J-\sum_j\gamma_j+\alpha-1}(1-p)^{\sum_j\gamma_j}$. Hence, the integral inside the summation in (\ref{eqn:jointInt}) is
    \begin{eqnarray*}
    \int_p \int_{\boldsymbol{\Theta}} f(\mathbf{y},\boldsymbol{\Theta},\boldsymbol{\gamma},p)d\boldsymbol{\Theta} dp &\propto& \int_{\boldsymbol{\Theta}}f(\mathbf{y},\boldsymbol{\Theta}\mid \boldsymbol{\gamma})d\boldsymbol{\Theta},
    \end{eqnarray*}
    since $\int_0^1 p^{J-\sum_j\gamma_j+\alpha-1}(1-p)^{\sum_j\gamma_j} dp < \infty$. Also, the summation over $\boldsymbol{\gamma}$ has $2^J$ terms, so it suffices to establish that
    \[
    \int_{\rho}\int_{\tau^2}\int_{\sigma^2}\int_{\boldsymbol{\mu}} f(\mathbf{y},\boldsymbol{\mu},\sigma^2,\tau^2,\rho \mid \boldsymbol{\gamma}) d\boldsymbol{\mu}d\sigma^2d\tau^2d\rho < \infty, \hspace{3mm} \forall \boldsymbol{\gamma} \in \{0,1\}^J.
    \]

    We begin with the case when $\gamma_j = 1,$ for all $j \in \{1,\ldots,J\}$ and define $f_{\mathbf{1}}(\mathbf{y} \mid  \boldsymbol{\mu}) := f(\mathbf{y} \mid \boldsymbol{\mu}, \boldsymbol{\gamma} = \mathbf{1})$ to simplify notation. We have that
    \begin{eqnarray*}
    f_{\mathbf{1}}(\mathbf{y} \mid  \boldsymbol{\mu}) &=& \prod_{j=1}^Jf(y_j \mid \mu_j, \gamma_j = 1) = \prod_{j=1}^JN(y_j \mid \mu_j, \sigma^2).
    \end{eqnarray*}
    Now, since $\rho \in (\nu_1^{-1},\nu_J^{-1})$ implies that $(\mathbf{D}^{\ast}_w-\rho\mathbf{W}) > 0$, where the notation $\mathbf{A} > 0$ means the matrix $\mathbf{A}$ is positive definite \citep*{BanerjeeEtAl04}, the prior on $\boldsymbol{\mu} \sim N_J(\mathbf{0}, ~\tau^2(\mathbf{D}^{\ast}_w-\rho\mathbf{W})^{-1})$. We can thus integrate $f_{\mathbf{1}}(\mathbf{y} \mid  \boldsymbol{\mu})\pi_{\boldsymbol{\mu}}(\boldsymbol{\mu} \mid \tau^2, \rho)$ with respect to $\boldsymbol{\mu}$ as the convolution of two normal densities. The marginal density of $\mathbf{y}$ is then that of a $N_J(\mathbf{0}, ~\sigma^2\mathbf{I} + \tau^2(\mathbf{D}^{\ast}_w-\rho\mathbf{W})^{-1})$ distribution. So, we know the integration yields
    \begin{eqnarray*}
    f_{\mathbf{1}}(\mathbf{y}, \sigma^2, \tau^2, \rho) &=& \pi_{_{(\tau^2,\sigma^2)}}(\tau^2, \sigma^2)\pi_{\rho}(\rho)N_J(\mathbf{y} \mid \mathbf{0}, \sigma^2\mathbf{I} + \tau^2(\mathbf{D}_2^{\ast} - \rho\mathbf{W})^{-1}).
    \end{eqnarray*}
    Now, we reparameterize the variance components by defining $\eta := \tau^2/\sigma^2$ and integrate over $\sigma^2$ using the inverse gamma integral to obtain
    \begin{eqnarray*}
    f_{\mathbf{1}}(\mathbf{y}, \eta, \rho) &\propto& \pi_{\rho}(\rho)(1+\eta)^{-2}| \mathbf{I}+\eta(\mathbf{D}^{\ast}_w-\rho\mathbf{W})^{-1}|^{-1/2}\\
        & & \times \int_0^{\infty}(\sigma^2)^{-(J/2)-1}\exp\left(-\frac{1}{2\sigma^2}\mathbf{y}^T(\mathbf{I}+\eta(\mathbf{D}^{\ast}_w-\rho\mathbf{W})^{-1})^{-1}\mathbf{y}\right)d\sigma^2\\
        && \hfill \\
        &\propto& \pi_{\rho}(\rho)(1+\eta)^{-2}\frac{| \mathbf{I}+\eta(\mathbf{D}^{\ast}_w-\rho\mathbf{W})^{-1}|^{-1/2}}{\left(\mathbf{y}^T(\mathbf{I}+\eta(\mathbf{D}^{\ast}_w-\rho\mathbf{W})^{-1})^{-1}\mathbf{y}\right)^{J/2}}.
    \end{eqnarray*}
    We can rewrite the quadratic form in the denominator of the last expression as
    \begin{eqnarray*}
    \mathbf{y}^T(\mathbf{I}+\eta(\mathbf{D}^{\ast}_w-\rho\mathbf{W})^{-1})^{-1}\mathbf{y} &=& \mathbf{y}^T(\mathbf{D}^{\ast}_w)^{1/2}(\mathbf{D}^{\ast}_w + \eta(\mathbf{I}-\rho\mathbf{W}^{\ast})^{-1})^{-1}(\mathbf{D}^{\ast}_w)^{1/2}\mathbf{y} \\
     &=& \mathbf{x}^T(\mathbf{D}^{\ast}_w + \eta(\mathbf{I}-\rho\mathbf{W}^{\ast})^{-1})^{-1}\mathbf{x},
    \end{eqnarray*}
    where $\mathbf{W}^{\ast} = (\mathbf{D}^{\ast}_w)^{-1/2}\mathbf{W}(\mathbf{D}^{\ast}_w)^{-1/2}$ and $\mathbf{x} = (\mathbf{D}^{\ast}_w)^{1/2}\mathbf{y}$. Let $w_{_{(J)}} = \stackrel{\max}{_{_{1\leq j\leq J}}}w_j.$
    Then, after substantial simplification (see Supplementary Material) it can be shown that, for all $\rho \in (\nu_1^{-1},\nu_J^{-1})$,
    \begin{eqnarray}
    \mathbf{x}^T((w_{_{(J)}} + d)\mathbf{I} + \eta(\mathbf{I}-\rho\mathbf{W}^{\ast})^{-1})^{-1}\mathbf{x} &\geq& k(w_{_{(J)}}+ d +\eta)^{-1}\\
    \Rightarrow (\mathbf{x}^T((w_{_{(J)}} + d)\mathbf{I} + \eta(\mathbf{I}-\rho\mathbf{W}^{\ast})^{-1})^{-1}\mathbf{x})^{-J/2} &\leq& k^{\prime}(w_{_{(J)}}+ d +\eta)^{J/2}
    \end{eqnarray}\label{eqn:prfIneq1}
    where $0< k^{\prime} < \infty$ is constant.

    Similarly, after substantial simplification (see Supplementary Material), we can show that
    \begin{equation}
    |\mathbf{I} + \eta(\mathbf{D}^{\ast}_w -\rho\mathbf{W}^{\ast})^{-1}|^{-1/2} \leq \frac{k^{\prime}\prod_{j=1}^J\max\{(1-\rho\nu_j)^{1/2}, 1\}}{(w_{_{(1)}} + d + \eta)^{J/2}}.
    \end{equation}\label{eqn:prfIneq2}

    We have now established that the density function satisfies
    \begin{eqnarray*}
    f(\mathbf{y}, \eta, \rho) &\leq& C(w_{_{(J)}} + d + \eta)^{J/2}\left(\frac{\prod_{j=1}^J\max\{(1-\rho\nu_j)^{1/2}, 1\}}{(w_{_{(1)}} + d + \eta)^{J/2}}\right)(1+\eta)^{-2}\pi_{\rho}(\rho)\\
        &=& \begin{cases}
                C\left(\frac{w_{_{(J)}} + d +\eta}{w_{_{(1)}}+ d +\eta}\right)^{J/2}\left(\frac{\prod_{j=r_1+1}^J(1-\rho\nu_j)^{1/2}}{(1+\eta)^2}\right)\pi_{\rho}(\rho), & \rho < 0\\
                C\left(\frac{w_{_{(J)}}+ d +\eta}{w_{_{(1)}}+ d +\eta}\right)^{J/2}\left(\frac{\prod_{j=1}^{J-r_2}(1-\rho\nu_j)^{1/2}}{(1+\eta)^2}\right)\pi_{\rho}(\rho), & \rho > 0
            \end{cases}\\
    \end{eqnarray*}
    where $C$ is a finite positive constant.

    Now, $\pi_{\rho}(\rho) \propto I(\nu_1^{-1} < \rho < \nu_J^{-1})$,
    \[
    \int_{\nu_1^{-1}}^0 \prod_{j=r_1+1}^J(1-\rho\nu_j)^{1/2}d\rho < \int_{\nu_1^{-1}}^0 (1-\rho\nu_J)^{\frac{J-r_1}{2}}d\rho < \infty,
    \]
    and
    \[
    \int_0^{\nu_J^{-1}}\prod_{j=1}^{J-r_2}(1-\rho\nu_j)^{1/2}d\rho < \int_0^{\nu_J^{-1}}(1-\rho\nu_1)^{\frac{J-r_2}{2}}d\rho < \infty.
    \]
    Also,
    \[
    \left(\frac{w_{_{(J)}}+ d +\eta}{w_{_{(1)}}+ d +\eta}\right)^{J/2}(1+\eta)^{-2} \leq k^{\prime}(1+\eta)^{-2},
    \]
    for all $\eta$,
    \[
    \Rightarrow \int_0^{\infty}\left(\frac{w_{_{(J)}}+d+\eta}{w_{_{(1)}}+d+\eta}\right)^{J/2}(1+\eta)^{-2} d\eta < \infty.
    \]
    From this, it follows that the integral $\int_{\nu_1^{-1}}^{\nu_J^{-1}}\int_0^{\infty}f_{\mathbf{1}}(\mathbf{y},\eta,\rho)d\eta d\rho$ is finite, showing that $\int_{\boldsymbol{\Theta}}f(\mathbf{y},\boldsymbol{\Theta}\mid \boldsymbol{\gamma})d\boldsymbol{\Theta} < \infty$ when $\gamma_1 = \gamma_2 = \cdots = \gamma_J = 1$.

    We now turn our attention to the case $|\{\gamma_j: \gamma_j = 1\}| = J_1 < J$, where $|\cdot|$ denotes the cardinality of a set. Let $\boldsymbol{\gamma}^{\ast}$ denote such an arbitrary configuration of $\boldsymbol{\gamma}$. Further, let $S_1 = \{j: \gamma_j=1\} \subsetneq \{1,\ldots,J\}$, $\mathbf{y}_1 = \{y_j: j\in S_1\}$, and $\mathbf{y}_0 = \{y_j: j\in S_1^c\}$ so that we may partition $\mathbf{y}$ as $\mathbf{y} = (\mathbf{y}_0^T, \mathbf{y}_1^T)^T$. Our strategy here is to show $h(\mathbf{y}_0, \mathbf{y}_1) \equiv h(\mathbf{y}) := \int_{\boldsymbol{\Theta}}f(\mathbf{y}, \boldsymbol{\Theta} \mid \boldsymbol{\gamma}^{\ast})d\boldsymbol{\Theta} < \infty$ by showing that $\int_{\mathbf{y}_1}h(\mathbf{y}_0, \mathbf{y}_1)d\mathbf{y}_1 < \infty$ for all $S_1^c$ and $\mathbf{y}_0 ~(a.e.)$. The key elements in the argument are that $J_1 < J \Rightarrow J-J_1 > 0$ and that we can factor the double integral as
    \[
    \int_{\mathbb{R}^J}\int_{\mathbb{R}^J}f(\mathbf{y}\mid\boldsymbol{\mu})\pi(\boldsymbol{\mu})d\boldsymbol{\mu}d\mathbf{y} = \left(\int_{\mathbb{R}^J}f(\mathbf{y}\mid\boldsymbol{\mu})d\mathbf{y}\right)\left(\int_{\mathbb{R}^J}\pi(\boldsymbol{\mu})d\boldsymbol{\mu}\right).
    \]

    Now, we have
    \begin{eqnarray*}
    \int_{\mathbf{y}_1} h(\mathbf{y}_0, \mathbf{y}_1)d\mathbf{y}_1 &=& \int_{\sigma^2}\int_{\tau^2}\int_{\rho}\int_{\boldsymbol{\mu}}\int_{\mathbf{y}_1}f(\mathbf{y}_0, \mathbf{y}_1, \tau^2, \sigma^2, \rho, \boldsymbol{\mu} \mid \boldsymbol{\gamma}^{\ast})d\mathbf{y}_1 d\boldsymbol{\mu} d\rho d\tau^2 d\sigma^2\\
        &\propto& \int_0^{\infty}(\sigma^2)^{-\frac{J-J_1}{2}}\exp\left(-\frac{1}{2\sigma^2}\sum_{j\in S_1^c}y_j^2\right)\int_0^{\infty}(\sigma^2+\tau^2)^{-2} d\tau^2 d\sigma^2\\
        && \hfill \\
        &=& \int_0^{\infty}\left(\frac{1}{\sigma^2}\right)^{\frac{J-J_1}{2}+1}\exp\left[-\left(\frac{1}{\sigma^2}\right)\left(\frac{\sum_{j\in S_1^c}y_j^2}{2}\right)\right]d\sigma^2 < \infty,
    \end{eqnarray*}
    since the integral is that of the kernel of an inverse gamma density over $(0, \infty)$.

    The proof is completed by considering the case $\gamma_j = 0, \text{ for all } j$. The preceding argument still applies, though, with $S_1 = \emptyset$ and $J_1 = 0$ (so that integration over $\mathbf{y}_1$ is not done). The result is therefore established.

\end{appendices}

\bibliographystyle{asa}
\bibliography{DependentBayesianTestingRef}

\begin{thebibliography}{60}
\newcommand{\enquote}[1]{``#1''}
\expandafter\ifx\csname natexlab\endcsname\relax\def\natexlab#1{#1}\fi

\bibitem[{Banerjee et~al.(2015)Banerjee, Carlin, and Gelfand}]{BanerjeeEtAl04}
Banerjee, S., Carlin, B.~P., and Gelfand, A.~E. (2015), \textit{{Hierarchical
  Modeling and Analysis for Spatial Data}}, Boca Raton: Chapman \& Hall/CRC,
  2nd ed.

\bibitem[{Banks et~al.(2012)Banks, Datta, Karr, Lynch, Niemi, and
  Vera}]{BanksEtAl09}
Banks, D., Datta, G., Karr, A., Lynch, J., Niemi, J., and Vera, F. (2012),
  \enquote{{Bayesian CAR models for syndromic surveillance on multiple data
  streams: Theory and practice},} \textit{Information Fusion}, 13, 105--116.

\bibitem[{Barbieri and Berger(2004)}]{BarbieriBerger04}
Barbieri, M.~M. and Berger, J.~O. (2004), \enquote{{Optimal predictive model
  selection},} \textit{Annals of Statistics}, 32, 870--897.

\bibitem[{Benjamini(2010)}]{Benjamini2010}
Benjamini, Y. (2010), \enquote{{Discovering the false discovery rate},}
  \textit{Journal of the Royal Statistical Society, Series B}, 72, 405--416.

\bibitem[{Benjamini and Hochberg(1995)}]{BenjaminiHochberg95}
Benjamini, Y. and Hochberg, Y. (1995), \enquote{{Controlling the false
  discovery rate: A practical and powerful approach to multiple testing},}
  \textit{Journal of the Royal Statistical Society, Series B}, 57, 289--300.

\bibitem[{Besag(1974)}]{Besag74}
Besag, J. (1974), \enquote{{Spatial interaction and the statistical analysis of
  lattice systems},} \textit{Journal of the Royal Statistical Society, Series
  B}, 36, 192--236.

\bibitem[{Besag et~al.(1991)Besag, York, and Molli\'e}]{BesagEtAl91}
Besag, J., York, J.~C., and Molli\'e, A. (1991), \enquote{{Bayesian image
  restoration, with two applications in spatial statistics (with discussion)},}
  \textit{Annals of the Institute of Statisical Mathematics}, 43, 1--59.

\bibitem[{Bogdan et~al.(2008)Bogdan, Ghosh, and Tokdar}]{BogdanEtAl08}
Bogdan, M., Ghosh, J.~K., and Tokdar, S.~T. (2008), \enquote{{A comparison of
  the Benjamini-Hochberg procedure with some Bayesian rules for multiple
  testing},} in \textit{Beyond Parametrics in Interdisciplinary Research:
  Festschrift in Honor of Professor Pranab K. Sen}, eds. Balakrishnan, N.,
  Pe{\~{n}}a, E.~A., and Silvapulle, M.~J., Beachwood: Institute of
  Mathematical Statistics, pp. 211--230.

\bibitem[{Brown et~al.(2014)Brown, Lazar, Datta, Jang, and
  McDowell}]{BrownEtAl14}
Brown, D.~A., Lazar, N.~A., Datta, G.~S., Jang, W., and McDowell, J.~E. (2014),
  \enquote{{Incorporating spatial dependence into Bayesian multiple testing of
  statistical parametric maps in functional neuroimaging},}
  \textit{NeuroImage}, 84, 97--112.

\bibitem[{Carlin and Banerjee(2003)}]{CarlinBanerjee03}
Carlin, B. and Banerjee, S. (2003), \enquote{{Hierarchical multivariate CAR
  models for spatio-temporally correlated survival data},} in \textit{Bayesian
  Statistics 7}, eds. Bernardo, J.~M., Bayarri, M.~J., Berger, J.~O., Dawid,
  A.~P., Heckerman, D., Smith, A. F.~M., and West, M., Oxford: Oxford
  University Press, pp. 45--63.

\bibitem[{Carlin and Louis(2009)}]{CarlinLouis09}
Carlin, B.~P. and Louis, T.~A. (2009), \textit{{Bayesian Methods for Data
  Analysis}}, Boca Raton: Chapman \& Hall/CRC, 3rd ed.

\bibitem[{Carvalho et~al.(2008)Carvalho, Chang, Lucas, Nevins, Wang, and
  West}]{CarvalhoEtAl08}
Carvalho, C.~M., Chang, J., Lucas, J.~E., Nevins, J.~R., Wang, Q., and West, M.
  (2008), \enquote{{High-dimensional sparse factor modeling: Applications in
  gene expression genomics},} \textit{Journal of the American Statistical
  Association}, 103, 1438--1456.

\bibitem[{Clyde and George(2004)}]{ClydeGeorge04}
Clyde, M. and George, E.~I. (2004), \enquote{{Model uncertainty},}
  \textit{Statistical Science}, 19, 81--94.

\bibitem[{Comon(1994)}]{Comon94}
Comon, P. (1994), \enquote{{Independent component analysis, A new concept?}}
  \textit{Signal Processing}, 36, 287--314.

\bibitem[{Dryden et~al.(2009)Dryden, Koloydenko, and Zhou}]{DrydenKoloZhou09}
Dryden, I.~L., Koloydenko, A., and Zhou, D. (2009), \enquote{{Non-Euclidean
  statistics for covariance matrices, with applications to diffusion tensor
  imaging},} \textit{Annals of Applied Statistics}, 3, 1102--1123.

\bibitem[{Eberly and Carlin(2000)}]{EberlyCarlin00}
Eberly, L.~E. and Carlin, B.~P. (2000), \enquote{{Identifiability and
  convergence issues for Markov chain Monte Carlo fitting of spatial models},}
  \textit{Statistics in Medicine}, 19, 2279--2294.

\bibitem[{Efron(2007)}]{Efron07}
Efron, B. (2007), \enquote{{Correlation and large-scale simultaneous
  significance testing},} \textit{Journal of the American Statistical
  Association}, 102, 93--103.

\bibitem[{Efron(2008)}]{Efron08}
--- (2008), \enquote{{Microarrays, empirical Bayes, and the two-groups model},}
  \textit{Statistical Science}, 23, 1--22.

\bibitem[{Efron(2010)}]{Efron10}
--- (2010), \textit{{Large Scale Inference: Empirical Bayes Methods for
  Estimation, Testing, and Prediction}}, New York: Cambridge University Press.

\bibitem[{Efron and Tibshirani(2002)}]{EfronTib02}
Efron, B. and Tibshirani, R. (2002), \enquote{{Empirical Bayes methods and
  false discovery rates for microarrays},} \textit{Genetic Epidemiology}, 23,
  70--86.

\bibitem[{Efron and Tibshirani(2007)}]{EfronTib07}
--- (2007), \enquote{{On testing the significance of sets of genes},}
  \textit{Annals of Applied Statistics}, 1, 107--129.

\bibitem[{Efron et~al.(2001)Efron, Tibshirani, Storey, and
  Tusher}]{EfronEtAl01}
Efron, B., Tibshirani, R., Storey, J.~D., and Tusher, V. (2001),
  \enquote{{Empirical Bayes analysis of a microarray experiment},}
  \textit{Journal of the American Statistical Association}, 96, 1151--1160.

\bibitem[{Friston et~al.(1995)Friston, Holmes, Worsley, Poline, Frith, and
  Frackowiak}]{FristonEtAl95}
Friston, K.~J., Holmes, A.~P., Worsley, K.~J., Poline, J.-P., Frith, C.~D., and
  Frackowiak, R. S.~J. (1995), \enquote{{Statistical parametric maps in
  functional imaging: A general linear approach},} \textit{Human Brain
  Mapping}, 2, 189--210.

\bibitem[{Gelfand and Sahu(1999)}]{GelfandSahu99}
Gelfand, A.~E. and Sahu, S.~K. (1999), \enquote{{Identifiability, improper
  priors, and Gibbs sampling for generalized linear models},} \textit{Journal
  of the American Statistical Association}, 94, 247--253.

\bibitem[{Gelfand and Smith(1990)}]{GelfandSmith90}
Gelfand, A.~E. and Smith, A. F.~M. (1990), \enquote{{Sampling-based approaches
  to calculating marginal densities},} \textit{Journal of the American
  Statistical Association}, 85, 398--409.

\bibitem[{Gelman(2006)}]{Gelman06}
Gelman, A. (2006), \enquote{{Prior distributions for variance parameters in
  hierarchical models},} \textit{Bayesian Analysis}, 1, 515--533.

\bibitem[{Gelman et~al.(2014)Gelman, Carlin, Stern, Dunson, Vehtari, and
  Rubin}]{GelmanEtAl14}
Gelman, A., Carlin, J.~B., Stern, H.~S., Dunson, D.~B., Vehtari, A., and Rubin,
  D.~B. (2014), \textit{{Bayesian Data Analysis}}, Boca Raton: Chapman \&
  Hall/CRC, 3rd ed.

\bibitem[{Gelman and Rubin(1992)}]{GelmanRubin92}
Gelman, A. and Rubin, D.~B. (1992), \enquote{{Inference from iterative
  simulation using multiple sequences (with discussion)},} \textit{Statistical
  Science}, 7, 457--511.

\bibitem[{Geman and Geman(1984)}]{GemanGeman84}
Geman, S. and Geman, D. (1984), \enquote{{Stochastic relaxation, Gibbs
  distributions and the Bayesian restoration of images},} \textit{IEEE
  Transactions on Pattern Analysis and Machine Intelligence}, 6, 721--741.

\bibitem[{George and McCulloch(1993)}]{GeorgeMcCulloch93}
George, E.~I. and McCulloch, R.~E. (1993), \enquote{{Variable selection via
  Gibbs sampling},} \textit{Journal of the American Statistical Association},
  88, 881--889.

\bibitem[{Geweke(1996)}]{Geweke96}
Geweke, J. (1996), \enquote{{Variable selection and model comparison in
  regression},} in \textit{Bayesian Statistics 5}, eds. Bernardo, J.~M.,
  Berger, J.~O., Dawid, A.~P., and Smith, A. F.~M., Oxford: Oxford University
  Press, pp. 609--620.

\bibitem[{Gilks et~al.(1996)Gilks, Richardson, and Spiegelhalter}]{GilksEtAl96}
Gilks, W., Richardson, S., and Spiegelhalter, D. (eds.) (1996), \textit{{Markov
  Chain Monte Carlo in Practice}}, London: Chapman \& Hall.

\bibitem[{Gilks and Wild(1992)}]{GilksWild92}
Gilks, W.~R. and Wild, P. (1992), \enquote{{Adaptive rejection sampling for
  Gibbs sampling},} \textit{Journal of the Royal Statistical Society, Series
  C}, 41, 337--348.

\bibitem[{Haario et~al.(2005)Haario, Saksman, and Tamminen}]{HaarioEtAl05}
Haario, H., Saksman, E., and Tamminen, J. (2005), \enquote{{Componentwise
  adaptation for high dimensional MCMC},} \textit{Computational Statistics},
  20, 265--273.

\bibitem[{Higdon(1994)}]{Higdon94}
Higdon, D. (1994), \enquote{{Spatial applications of Markov chain Monte Carlo
  for Bayesian inference},} Unpublished doctoral thesis, University of
  Washington, Department of Statistics.

\bibitem[{Kass et~al.(1998)Kass, Carlin, Gelman, and Neal}]{KassEtAl98}
Kass, R.~E., Carlin, B.~P., Gelman, A., and Neal, R. (1998), \enquote{{Markov
  chain Monte Carlo in practice: A roundtable discussion},} \textit{The
  American Statistician}, 52, 93--100.

\bibitem[{Lee et~al.(2014)Lee, Jones, Caffo, and Bassett}]{LeeEtAl14}
Lee, K.-J., Jones, G.~L., Caffo, B.~S., and Bassett, S.~S. (2014),
  \enquote{{Spatial Bayesian variable selection models on functional magnetic
  resonance imaging time-series data},} \textit{Bayesian Analysis}, 9,
  699--732.

\bibitem[{Li and Zhang(2010)}]{LiZhang10}
Li, F. and Zhang, N.~R. (2010), \enquote{{Bayesian variable selection in
  structured high-dimensional covariate spaces with applications in genomics},}
  \textit{Journal of the American Statistical Association}, 105, 1202--1214.

\bibitem[{Li and Tibshirani(2011)}]{LiTibshirani11}
Li, J. and Tibshirani, R. (2011), \enquote{{Finding consistent patterns: A
  nonparametric approach for identifying differential expression in RNA-Seq
  data},} \textit{Statistical Methods in Medical Research}, 22, 519--536.

\bibitem[{Love et~al.(2014)Love, Huber, and Anders}]{LoveEtAl14}
Love, M.~I., Huber, W., and Anders, S. (2014), \enquote{{Moderated estimation
  of fold change and dispersion for RNA-Seq data with DESeq2},} BioRxiv
  Preprint doi: http://dx.doi.org/10.1101/002832.

\bibitem[{Macnab(1992)}]{Macnab92}
Macnab, R.~M. (1992), \enquote{{Genetics and biogenesis of bacterial
  flagella},} \textit{Annual Review of Genetics}, 26, 131--158.

\bibitem[{McLachlan and Peel(2000)}]{McLachlanPeel00}
McLachlan, G. and Peel, D. (2000), \textit{{Finite Mixture Models}}, New York:
  John Wiley and Sons.

\bibitem[{Mitchell and Beauchamp(1988)}]{MitchellBeauchamp88}
Mitchell, T.~J. and Beauchamp, J.~J. (1988), \enquote{{Bayesian variable
  selection in linear regression},} \textit{Journal of the American Statistical
  Association}, 83, 1023--1032.

\bibitem[{Neal(2003)}]{Neal03}
Neal, R.~M. (2003), \enquote{{Slice sampling},} \textit{Annals of Statistics},
  31, 705--741.

\bibitem[{O'Hara and Sillanp\"{a}\"{a}(2009)}]{OHaraSillanpaa09}
O'Hara, R.~B. and Sillanp\"{a}\"{a}, M.~J. (2009), \enquote{{A review of
  Bayesian variable selection methods: What, how, and which},} \textit{Bayesian
  Analysis}, 4, 85--118.

\bibitem[{Qiu et~al.(2005)Qiu, Klebanov, and Yakovlev}]{QiuEtAl05}
Qiu, X., Klebanov, L., and Yakovlev, A. (2005), \enquote{{Correlation between
  gene expression levels and limitations of the empirical Bayes methodology for
  finding differentially expressed genes},} \textit{Statistical Applications in
  Genetics and Molecular Biology}, 4, article 34.

\bibitem[{{R Core Team}(2015)}]{RCore}
{R Core Team} (2015), \textit{R: A Language and Environment for Statistical
  Computing}, R Foundation for Statistical Computing, Vienna, Austria.

\bibitem[{Scott and Berger(2006)}]{ScottBerger06}
Scott, J.~G. and Berger, J.~O. (2006), \enquote{{An exploration of aspects of
  Bayesian multiple testing},} \textit{Journal of Statistical Planning and
  Inference}, 136, 2144--2162.

\bibitem[{Scott and Berger(2010)}]{ScottBerger10}
--- (2010), \enquote{{Bayes and empirical-Bayes multiplicity adjustment in the
  variable-selection problem},} \textit{Annals of Statistics}, 38, 2587--2619.

\bibitem[{Serre(2002)}]{Serre02}
Serre, D. (2002), \textit{{Matrices: Theory and Applications}}, New York:
  Springer-Verlag.

\bibitem[{Smith and Fahrmeir(2007)}]{SmithFahr07}
Smith, M. and Fahrmeir, L. (2007), \enquote{{Spatial Bayesian variable
  selection with application to functional magnetic resonance imaging},}
  \textit{Journal of the American Statistical Association}, 102, 417--431.

\bibitem[{Stingo et~al.(2011)Stingo, Chen, Tadesse, and
  Vannucci}]{StingoEtAl11}
Stingo, F.~C., Chen, Y.~A., Tadesse, M.~G., and Vannucci, M. (2011),
  \enquote{{Incorporating biological information in the linear models: A
  Bayesian approach to the selection of pathways and genes},} \textit{Annals of
  Applied Statistics}, 5, 1978--2002.

\bibitem[{Subramanian et~al.(2005)Subramanian, Tamayo, Mootha, Mukherjee,
  Ebert, Gillette, Paulovich, Pomeroy, Golub, Lander, and
  Mesirov}]{SubramanianEtAl05}
Subramanian, A., Tamayo, P., Mootha, V.~K., Mukherjee, S., Ebert, B.~L.,
  Gillette, M.~A., Paulovich, A., Pomeroy, S.~L., Golub, T.~R., Lander, E.~S.,
  and Mesirov, J.~P. (2005), \enquote{{Gene set enrichment analysis: A
  knowledge-based approach for interpreting genome-wide expression profiles},}
  \textit{Proceedings of the National Academy of Sciences}, 102, 15545--15550.

\bibitem[{Tadesse et~al.(2005)Tadesse, Ibrahim, Vannucci, and
  Gentleman}]{TadesseEtAl05}
Tadesse, M., Ibrahim, J.~G., Vannucci, M., and Gentleman, R. (2005),
  \enquote{{Wavelet thresholding with Bayesian false discovery rate control},}
  \textit{Biometrics}, 61, 25--35.

\bibitem[{Tusher et~al.(2001)Tusher, Tibshirani, and Chu}]{TusherEtAl01}
Tusher, V.~G., Tibshirani, R., and Chu, G. (2001), \enquote{{Significance
  analysis of microarrays applied to the ionizing radiation response},}
  \textit{Proceedings of the National Academy of Sciences}, 98, 5116--5121.

\bibitem[{Wantanabe(2010)}]{Wantanabe10}
Wantanabe, S. (2010), \enquote{{Asymptotic equivalence of Bayes cross
  validation and widely applicable information criterion in singular learning
  theory},} \textit{Journal of Machine Learning Research}, 11, 3571--3594.

\bibitem[{West(2003)}]{West03}
West, M. (2003), \enquote{{Bayesian factor regression models in the `large p,
  small n' paradigm},} in \textit{Bayesian Statistics 7}, eds. Bernardo, J.~M.,
  Bayarri, M.~J., Berger, J.~O., Dawid, A.~P., Heckerman, D., Smith, A. F.~M.,
  and West, M., Oxford: Oxford University Press, pp. 723--732.

\bibitem[{Xiao et~al.(2009)Xiao, Reilly, and Khodursky}]{XiaoEtAl09}
Xiao, G., Reilly, C., and Khodursky, A.~B. (2009), \enquote{{Improved detection
  of differentially expressed genes through incorporation of gene locations},}
  \textit{Biometrics}, 65, 805--814.

\bibitem[{Zhang et~al.(2014)Zhang, Guindani, Versace, and
  Vannucci}]{ZhangEtAl14}
Zhang, L., Guindani, M., Versace, F., and Vannucci, M. (2014), \enquote{{A
  spatio-temporal nonparametric Bayesian variable selection model of fMRI data
  for clustering correlated time courses},} \textit{NeuroImage}, 95, 162--175.

\bibitem[{Zhao et~al.(2014)Zhao, Kang, and Yu}]{ZhaoEtAl14}
Zhao, Y., Kang, J., and Yu, T. (2014), \enquote{{A Bayesian nonparametric
  mixutre model for selecting genes and gene subnetworks},} \textit{Annals of
  Applied Statistics}, 8, 999--1021.

\end{thebibliography}

\newpage

\begin{center}
    {\huge A Bayesian Generalized CAR Model for Correlated Signal Detection \\ Supplementary Material}\\[18pt]

    {\large D. Andrew Brown \hspace{18pt} Gauri S. Datta \hspace{18pt} Nicole A. Lazar}
\end{center}

\section{Full Conditional Distributions for the Proposed Model}

Model (2.4) leads to the following full conditional distributions needed for a Gibbs sampling algorithm:
\begin{equation}
    \begin{aligned}
        \mu_j \mid \boldsymbol{\mu}_{(-j)}, \boldsymbol{\gamma}, \sigma^2, \eta, \rho, p, \mathbf{y} &\sim N\left(\frac{\gamma_jy_j + \rho\eta^{-1}\sum_{i \neq j}w_{ji}\mu_i}{\gamma_j + \eta^{-1}(w_j. + d)}, ~\sigma^2(\gamma_j + \eta^{-1}(w_j. + d))^{-1}\right), ~~j=1, \ldots, J\\
        \gamma_j \mid \boldsymbol{\mu}, \boldsymbol{\gamma}_{(-j)}, \sigma^2, \eta, \rho, p, \mathbf{y} &\sim \text{Bern}\left(\frac{(1-p)\exp\left(-(y_j - \mu_j)^2/2\sigma^2\right)}{(1-p)\exp\left(-(y_j - \mu_j)^2/2\sigma^2\right) + p\exp\left(-y_j^2/2\sigma^2\right)}\right), ~j= 1, \ldots, J\\
        \sigma^2 \mid \boldsymbol{\mu}, \boldsymbol{\gamma}, \eta, \rho, p, \mathbf{y} &\sim \text{InvGam}\left(J, ~\frac{(\mathbf{y} - \boldsymbol{\Gamma\mu})^T(\mathbf{y} - \boldsymbol{\Gamma\mu}) + \eta^{-1}\boldsymbol{\mu}^T(\mathbf{D}^{\ast}_w - \rho\mathbf{W})\boldsymbol{\mu}}{2}\right)\\
        p \mid \boldsymbol{\mu}, \boldsymbol{\gamma}, \sigma^2, \eta, \rho, \mathbf{y} &\sim \text{Beta}\left(J - \sum_{i=1}^J\gamma_i + \alpha, \sum_{i=1}^J\gamma_i + 1\right)\\
        [\eta \mid \boldsymbol{\mu}, \boldsymbol{\gamma}, \sigma^2, \rho, p, \mathbf{y}] &\propto \eta^{-J/2 - 2}\exp\left(-\frac{1}{\eta}\left(\frac{\boldsymbol{\mu}^T(\mathbf{D}^{\ast}_w - \rho\mathbf{W})\boldsymbol{\mu}}{2\sigma^2}\right)\right)\left(\frac{\eta}{1+\eta}\right)^2, ~\eta > 0\\
        [\rho \mid \boldsymbol{\mu}, \boldsymbol{\gamma}, \sigma^2, \eta, p, \mathbf{y}] &\propto |\mathbf{D}^{\ast}_w - \rho\mathbf{W}|^{1/2}\exp\left(-\frac{\boldsymbol{\mu}^T(\mathbf{D}^{\ast}_w - \rho\mathbf{W})\boldsymbol{\mu}}{2\eta\sigma^2}\right)I(\nu_1^{-1} < \rho < \nu_J^{-1}),\\
    \end{aligned}\label{eqn:fullConditionals}
\end{equation}
where $[~\cdot \mid \cdot~]$ denotes a conditional density.\\

\section{Additional Facts}\label{sec:PropProof}

%
%
In this Section, we provide a proof that the maximum value of $\rho$ is an increasing function of $d$ in Model (4). Also, we detail the simplifications used to obtain (7), (8), and (9)  in the proof. The simplifications make use of a simple inequality, which we state as a Lemma.

\subsection{Relationship Between $\rho$ and $d$}
Consider two values $d_1 < d_2$. Let $\mathbf{x}_i^{\ast} = \argmax_{\mathbf{x}} \mathbf{x}^T(\mathbf{D}_w + d_i\mathbf{I})^{-1/2}\mathbf{W}(\mathbf{D}_w + d_i\mathbf{I})^{-1/2}\mathbf{x} / \mathbf{x}^T\mathbf{x}, ~i= 1, 2$, and let $\lambda_{J,i} > 0$ be the maximum eigenvalue of $(\mathbf{D}_w + d_i\mathbf{I})^{-1/2}\mathbf{W}(\mathbf{D}_w + d_i\mathbf{I})^{-1/2}, ~i= 1,2$. Then it follows from the properties of the Rayleigh quotient that
\begin{eqnarray*}
    \lambda_{J,2} &=& \frac{\mathbf{x}_2^{\ast,T}(\mathbf{D}_w + d_2\mathbf{I})^{-1/2}\mathbf{W}(\mathbf{D}_w + d_2\mathbf{I})^{-1/2}\mathbf{x}_2^{\ast}}{\mathbf{x}_2^{\ast,T}\mathbf{x}_2^{\ast}} \\
        &=& \frac{1}{\mathbf{x}_2^{\ast,T}\mathbf{x}_2^{\ast}}\sum_i\sum_j \frac{w_{ij}x_{i,2}^{\ast}x_{j,2}^{\ast}}{\sqrt{w_i. + d_2}\sqrt{w_j. + d_2}} \\
        &<& \frac{1}{\mathbf{x}_2^{\ast,T}\mathbf{x}_2^{\ast}}\sum_i\sum_j \frac{w_{ij}x_{i,2}^{\ast}x_{j,2}^{\ast}}{\sqrt{w_i. + d_1}\sqrt{w_j. + d_1}} \\
        &\leq& \frac{\mathbf{x}_1^{\ast,T}(\mathbf{D}_w + d_1\mathbf{I})^{-1/2}\mathbf{W}(\mathbf{D}_w + d_1\mathbf{I})^{-1/2}\mathbf{x}_1^{\ast}}{\mathbf{x}_1^{\ast,T}\mathbf{x}_1^{\ast}} \\
        &=& \lambda_{J,1},
\end{eqnarray*}
so $\lambda_{J,1}^{-1} < \lambda_{J,2}^{-1}$. The inequality in the third line follows from the Perron-Frobenius Theorem \citep*[e.g.,][]{Serre02}, which guarantees that the elements of the eigenvector associated to $\lambda_{J,2}$ are all positive. Thus, the maximum possible value for $\rho$ is an increasing function in $d$.

\subsection{Facts Used in the Proof of Posterior Propriety}
\begin{lemma}\label{bLemma}
For any positive constants $a,b, \text{and } c$,
\[
\frac{b}{ba+c} < \frac{\max\{b,1\}}{a+c}.
\]
\end{lemma}
\begin{proof}
    Let $a,c \in \mathbb{R}^+$. Then, for $0 < b < 1$,
    \begin{eqnarray*}
    \frac{b}{ba+c} - \frac{1}{a+c} &=& \frac{b(a+c) - ba - c}{(ba+c)(a+c)}\\
     &=& \frac{c(b-1)}{(ba+c)(a+c)}\\
     &<& 0,
     \end{eqnarray*}
    and for $b > 1$,
    \begin{eqnarray*}
    \frac{b}{ba+c} - \frac{b}{a+c} &=& \frac{b(a+c) - b(ba + c)}{(ba+c)(a+c)}\\
     &=& \frac{ba(1-b)}{(ba+c)(a+c)}\\
     &<& 0.
    \end{eqnarray*}
\end{proof}

    Now, the matrix $(w_{_{(J)}} + d)\mathbf{I} - \mathbf{D}^{\ast}_w = w_{_{(J)}}\mathbf{I} - \mathbf{D}_w$ is diagonal with nonnegative entries and thus positive semidefinite, which we denote as $(w_{_{(J)}} + d)\mathbf{I} - \mathbf{D}^{\ast}_w \geq 0$.

    Let $\nu_1 \leq \nu_2 \leq \cdots \leq \nu_J$ be the ordered eigenvalues of $\mathbf{W}^{\ast}$. Then the eigenvalues of $\mathbf{I}-\rho\mathbf{W}^{\ast}$ are $1-\rho\nu_j, ~j = 1,\ldots,J$. Hence,
    \begin{eqnarray*}
    \rho \in (\nu_1^{-1},\nu_J^{-1}) &\Rightarrow& 1 - \rho\nu_j > 0, ~\forall j\\
      &\Rightarrow& (\mathbf{I}-\rho\mathbf{W}^{\ast}) > 0 \\
      &\Rightarrow& \eta(\mathbf{I}-\rho\mathbf{W}^{\ast})^{-1} > 0.
    \end{eqnarray*}
    By adding and subtracting $\eta(\mathbf{I}-\rho\mathbf{W}^{\ast})^{-1}$, we obtain
    \[
    (w_{_{(J)}} + d)\mathbf{I} - \mathbf{D}^{\ast}_w = (w_{_{(J)}} + d)\mathbf{I} + \eta(\mathbf{I}-\rho\mathbf{W}^{\ast})^{-1} - (\mathbf{D}^{\ast}_w + \eta(\mathbf{I}-\rho\mathbf{W}^{\ast})^{-1}) \geq 0.
    \]
    Making use of the fact that $\mathbf{B} > 0, \mathbf{A}-\mathbf{B} \geq 0 \Rightarrow \mathbf{B}^{-1} - \mathbf{A}^{-1} \geq 0$, it follows that
    \[
    (\mathbf{D}^{\ast}_w + \eta(\mathbf{I}-\rho\mathbf{W}^{\ast})^{-1})^{-1} - ((w_{_{(J)}} + d)\mathbf{I} + \eta(\mathbf{I}-\rho\mathbf{W}^{\ast})^{-1})^{-1} \geq 0
    \]
    \[
     \Rightarrow \mathbf{x}^T(\mathbf{D}^{\ast}_w + \eta(\mathbf{I}-\rho\mathbf{W}^{\ast})^{-1})^{-1}\mathbf{x} \geq \mathbf{x}^T((w_{_{(J)}} + d)\mathbf{I} + \eta(\mathbf{I}-\rho\mathbf{W}^{\ast})^{-1})^{-1}\mathbf{x}
    \]
    \[
    \Rightarrow (\mathbf{x}^T(\mathbf{D}^{\ast}_w + \eta(\mathbf{I}-\rho\mathbf{W}^{\ast})^{-1})^{-1}\mathbf{x})^{-J/2} \leq (\mathbf{x}^T((w_{_{(J)}} + d)\mathbf{I} + \eta(\mathbf{I}-\rho\mathbf{W}^{\ast})^{-1})^{-1}\mathbf{x})^{-J/2}.
    \]

    Now, $\mathbf{W}^{\ast}$ is symmetric, so it has a spectral decomposition of the form $\mathbf{W}^{\ast} = \mathbf{P}\mathbf{M}\mathbf{P}^T$, where $\mathbf{P}$ is the orthogonal matrix of eigenvectors of $\mathbf{W}^{\ast}$ and $\mathbf{M}$ is the diagonal matrix of eigenvalues. Let $\mathbf{u} = \mathbf{P}^T\mathbf{x} \Rightarrow \mathbf{x} = \mathbf{Pu}$ so that
    \begin{eqnarray*}
    \mathbf{x}^T((w_{_{(J)}} + d)\mathbf{I} + \eta(\mathbf{I}-\rho\mathbf{W}^{\ast})^{-1})^{-1}\mathbf{x} &=& \mathbf{u}^T\mathbf{P}^T((w_{_{(J)}} + d)\mathbf{I} + \eta(\mathbf{I}-\rho\mathbf{W}^{\ast})^{-1})^{-1}\mathbf{Pu} \\
        &=& \mathbf{u}^T(\mathbf{P}^T(w_{_{(J)}} + d)\mathbf{P} + \eta\mathbf{P}^T(\mathbf{I}-\rho\mathbf{W}^{\ast})^{-1}\mathbf{P})^{-1}\mathbf{u} \\
        &=& \mathbf{u}^T((w_{_{(J)}} + d)\mathbf{I} + \eta(\mathbf{I}-\rho\mathbf{M})^{-1})^{-1}\mathbf{u}\\
        &=& \sum_{j=1}^J \frac{(1-\rho\nu_j)u_j^2}{(w_{_{(J)}} + d)(1-\rho\nu_j)+\eta}.
    \end{eqnarray*}

    Since $\mathbf{D}^{\ast}_w$ is diagonal and the diagonal elements of $\mathbf{W}$ are zero, the diagonal elements of $\mathbf{W}^{\ast} = (\mathbf{D}^{\ast}_w)^{-1/2}\mathbf{W}(\mathbf{D}^{\ast}_w)^{-1/2}$ are zero and thus $\text{tr}(\mathbf{W}^{\ast}) = 0 = \sum_{j=1}^J \nu_j$. It must then be true that there are $r_1 > 0$ negative eigenvalues and $r_2 > 0$ positive eigenvalues of $\mathbf{W}^{\ast}$, since $r := r_1 + r_2 = \text{rank}(\mathbf{W}^{\ast}) > 0$. The summation in the last line can then be separated according to the sign of the eigenvalue in each term as
    \begin{equation}\label{eqn:posNeg}
    \underbrace{\sum_{j=1}^{r_1} \frac{(1-\rho\nu_j)u_j^2}{(w_{_{(J)}} + d)(1-\rho\nu_j)+\eta}}_{\nu_j < 0} + \underbrace{\sum_{j=J-r_2+1}^J \frac{(1-\rho\nu_j)u_j^2}{(w_{_{(J)}} + d)(1-\rho\nu_j)+\eta}}_{\nu_j > 0} + \underbrace{\sum_{j=r_1+1}^{J-r_2} \frac{u_j^2}{w_{_{(J)}} + d +\eta}}_{\nu_j = 0}.
    \end{equation}
    If $\nu_1^{-1} < \rho < 0$, then $0 < 1-\rho\nu_j < 1$ for $j=1,\ldots,r_1$ and $1-\rho\nu_j > 1$ for $j= J-r_2+1,\ldots,J,$ so
    \begin{eqnarray*}
    (\ref{eqn:posNeg}) &\geq& \sum_{j=J-r_2+1}^J \frac{(1-\rho\nu_j)u_j^2}{(w_{_{(J)}} + d)(1-\rho\nu_j)+\eta} + \sum_{j=r_1+1}^{J-r_2} \frac{u_j^2}{w_{_{(J)}} + d +\eta}\\
        &\geq& \sum_{j=J-r_2+1}^J \frac{u_j^2}{w_{_{(J)}} + d +\eta} + \sum_{j=r_1+1}^{J-r_2} \frac{u_j^2}{w_{_{(J)}}+ d +\eta}\\
        &=& (w_{_{(J)}} + d +\eta)^{-1}\sum_{j=r_1+1}^J u_j^2,
    \end{eqnarray*}
    where the second line follows from noticing that $(w_{_{(J)}} + d +\eta)^{-1} - (1-\rho\nu_j)((w_{_{(J)}} + d)(1-\rho\nu_j)+\eta)^{-1} < 0$ for $\nu_j > 0$. Similarly, if $0 < \rho < \nu_J^{-1}$, then
    \begin{eqnarray*}
    (\ref{eqn:posNeg}) &\geq& \sum_{j=1}^{r_1} \frac{(1-\rho\nu_j)u_j^2}{(w_{_{(J)}} + d)(1-\rho\nu_j)+\eta} + \sum_{j=r_1+1}^{J-r_2} \frac{u_j^2}{w_{_{(J)}}+ d +\eta}\\
        &\geq& \sum_{j=1}^{r_1} \frac{u_j^2}{w_{_{(J)}}+ d +\eta} + \sum_{j=r_1+1}^{J-r_2} \frac{u_j^2}{w_{_{(J)}} + d +\eta}\\
        &=& (w_{_{(J)}} + d +\eta)^{-1}\sum_{j=1}^{J-r_2}u_j^2.
    \end{eqnarray*}

    Thus, we have that for all $\rho \in (\nu_1^{-1},\nu_J^{-1})$,
    \[
    \mathbf{x}^T((w_{_{(J)}} + d)\mathbf{I} + \eta(\mathbf{I}-\rho\mathbf{W}^{\ast})^{-1})^{-1}\mathbf{x} \geq k(w_{_{(J)}}+ d +\eta)^{-1}
    \]
    \[
    \Rightarrow (\mathbf{x}^T((w_{_{(J)}} + d)\mathbf{I} + \eta(\mathbf{I}-\rho\mathbf{W}^{\ast})^{-1})^{-1}\mathbf{x})^{-J/2} \leq k^{\prime}(w_{_{(J)}}+ d +\eta)^{J/2}
    \]
    where $0< k^{\prime} < \infty$ is constant. This establishes (7) and (8).

    To see that (9) holds, note that
    \begin{eqnarray*}
    | \mathbf{I} + \eta(\mathbf{D}^{\ast}_w-\rho\mathbf{W})^{-1} | &=& | (\mathbf{D}^{\ast}_w)^{-1/2}| | \mathbf{D}^{\ast}_w + \eta(\mathbf{I}-\rho\mathbf{W}^{\ast})^{-1}| | (\mathbf{D}^{\ast}_w)^{-1/2} | \\
        &=& |(\mathbf{D}^{\ast}_w)|^{-1} | \mathbf{D}^{\ast}_w + \eta(\mathbf{I}-\rho\mathbf{W}^{\ast})^{-1}| \\
        &\equiv& k| \mathbf{D}^{\ast}_w + \eta(\mathbf{I}-\rho\mathbf{W}^{\ast})^{-1}|.
    \end{eqnarray*}
    Letting $w_{_{(1)}} = \stackrel{\min}{_{_{1\leq j \leq J}}} w_j.$, and again adding and subtracting $\eta(\mathbf{I}-\rho\mathbf{W}^{\ast})^{-1}$, we obtain
    \[
    \mathbf{D}^{\ast}_w - (w_{_{(1)}} + d)\mathbf{I} \geq 0 \Rightarrow \mathbf{D}^{\ast}_w + \eta(\mathbf{I}-\rho\mathbf{W}^{\ast})^{-1} - ((w_{_{(1)}} + d)\mathbf{I} + \eta(\mathbf{I}-\rho\mathbf{W}^{\ast})^{-1}) \geq 0.
    \]
    But $\mathbf{B} > 0, ~\mathbf{A}-\mathbf{B} \geq 0 \text{ implies } |\mathbf{A}| \geq |\mathbf{B}|$, so we find that
    \[
    |\mathbf{D}^{\ast}_w + \eta(\mathbf{I}-\rho\mathbf{W}^{\ast})^{-1}| \geq |(w_{_{(1)}} + d)\mathbf{I} + \eta(\mathbf{I}-\rho\mathbf{W}^{\ast})^{-1}|
    \]
    \[
    \Rightarrow k| \mathbf{D}^{\ast}_w + \eta(\mathbf{I}-\rho\mathbf{W}^{\ast})^{-1} | \geq K| (w_{_{(1)}} + d)\mathbf{I} + \eta(\mathbf{I}-\rho\mathbf{W}^{\ast})^{-1}|.
    \]
    The eigenvalues of $(\mathbf{I}-\rho\mathbf{W}^{\ast})^{-1}$ are $(1-\rho\nu_j)^{-1}$, $j= 1,\ldots,J$, so it follows that the eigenvalues of $(w_{_{(1)}} + d)\mathbf{I} + \eta(\mathbf{I}-\rho\mathbf{W}^{\ast})^{-1}$ are $w_{_{(1)}} + d + \eta(1-\rho\nu_j)^{-1}$, $j= 1,\ldots,J$. Therefore,
    \begin{eqnarray*}
    k| (w_{_{(1)}} + d)\mathbf{I} + \eta(\mathbf{I}-\rho\mathbf{W}^{\ast})^{-1}| &=& k\prod_{j=1}^J(w_{_{(1)}} + d + \eta(1-\rho\nu_j)^{-1})\\
        &=& \frac{k\prod_{j=1}^J((1-\rho\nu_j)(w_{_{(1)}} + d) + \eta)}{\prod_{j=1}^J(1-\rho\nu_j)},
    \end{eqnarray*}
    and subsequently
    \[
    | \mathbf{I} + \eta(\mathbf{D}^{\ast}_w-\rho\mathbf{W})^{-1} |^{-1/2} \leq k^{\prime} \left(\frac{\prod_{j=1}^J((1-\rho\nu_j)(w_{_{(1)}} + d) + \eta)}{\prod_{j=1}^J(1-\rho\nu_j)}\right)^{-1/2},
    \]
    where $0 < k^{\prime} < \infty$ is constant. But $1-\rho\nu_j > 0, \text{ for all } j$, so by Lemma \ref{bLemma}
    \[
    \frac{1-\rho\nu_j}{(1-\rho\nu_j)(w_{_{(1)}} + d) + \eta} \leq \frac{\max\{1-\rho\nu_j, 1\}}{w_{_{(1)}} + d + \eta}, \hspace{3mm} \forall j \in \{1,\ldots,J\}
    \]
    \[
    \Rightarrow \frac{\prod_{j=1}^J(1-\rho\nu_j)}{\prod_{j=1}^J((1-\rho\nu_j)(w_{_{(1)}} + d) + \eta)} \leq \frac{\prod_{j=1}^J\max\{1-\rho\nu_j, 1\}}{(w_{_{(1)}} + d + \eta)^J}
    \]

\newpage

\section{Supplementary Figures}
\begin{figure}[h!]  
    \begin{center}
    \caption{Comparison of the $|t_2|$ prior on $\tau$ (Gelman, 2006) (after transforming to the $\tau^2$ scale) and the prior on $\tau^2 \mid \sigma^2$ suggested by Scott and Berger (2006), denoted by SB. Here, the scale parameter is set to 1 in both densities.}
        \includegraphics[scale= 0.4]{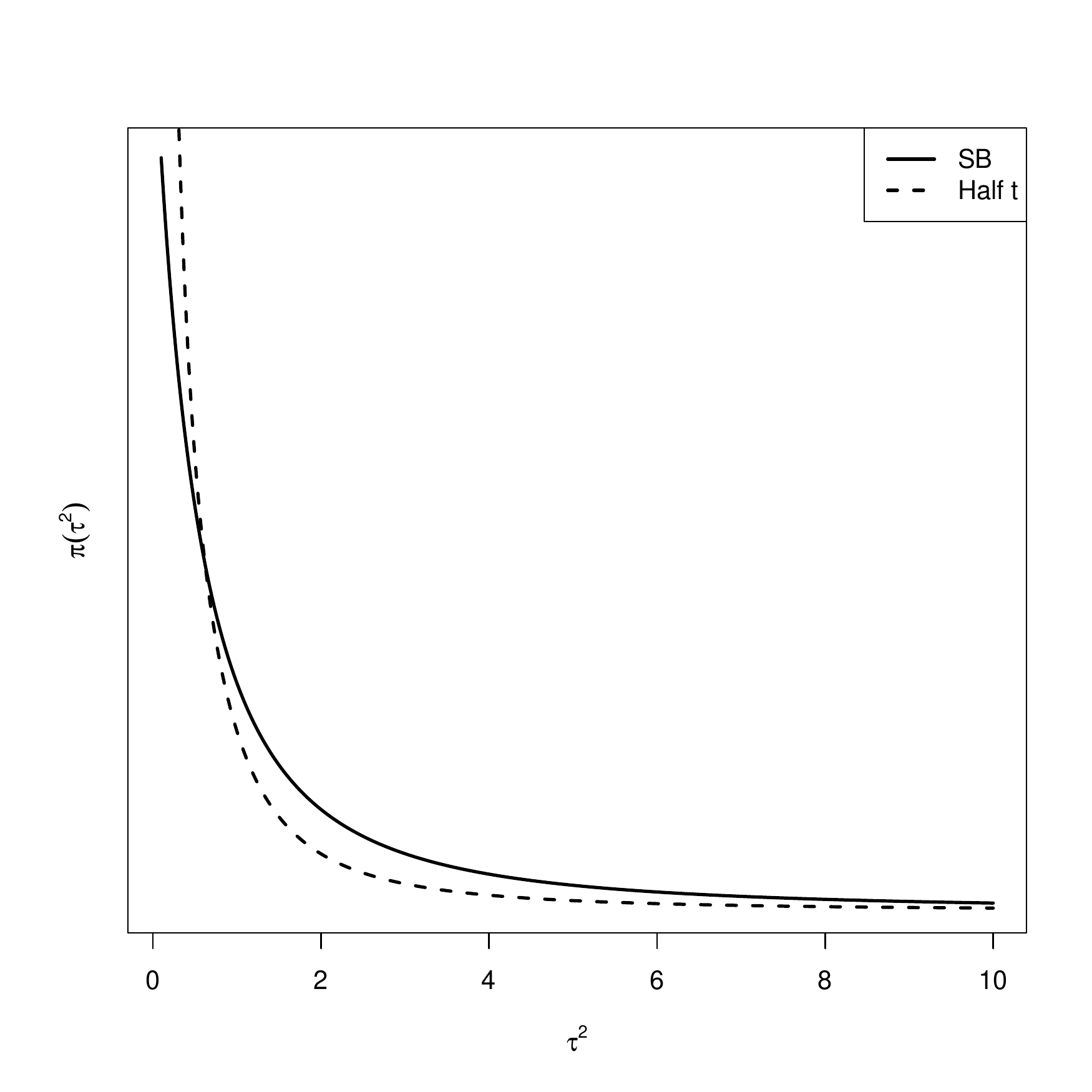}
        \label{fig:SBvHalfT}
    \end{center}
\end{figure}

\begin{figure}[h!tb]  
    \centering
    \caption{Simulated binary activation pattern drawn from an Ising distribution.}
        \includegraphics[scale= 0.4]{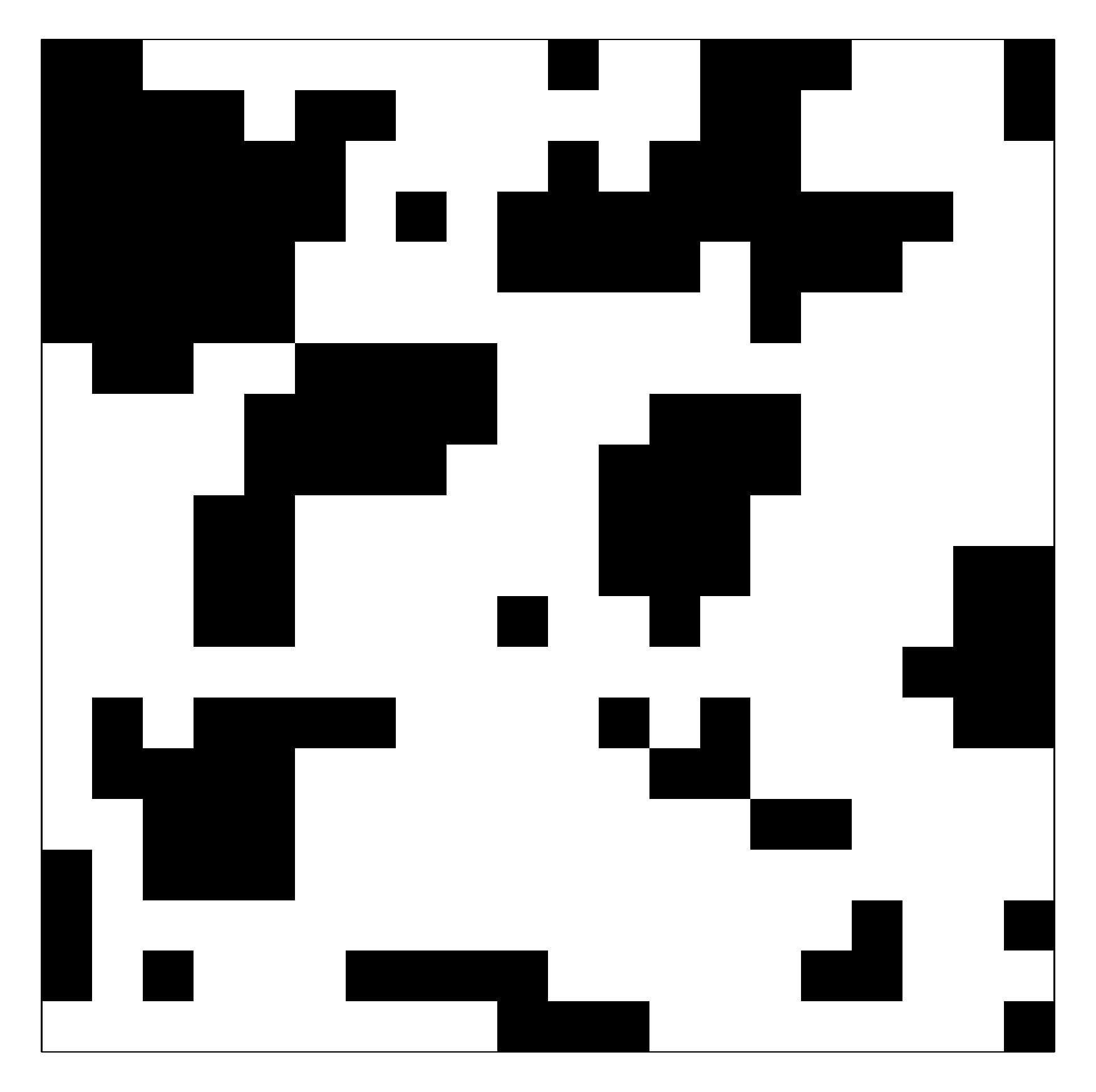}
        \label{fig:ActPattern}
\end{figure}

\begin{figure}[h!tb]  
    \centering
    \caption{Simulated data using the activation pattern in Figure \ref{fig:ActPattern} with non-null mean 3.5.}
        \includegraphics[scale= 0.4]{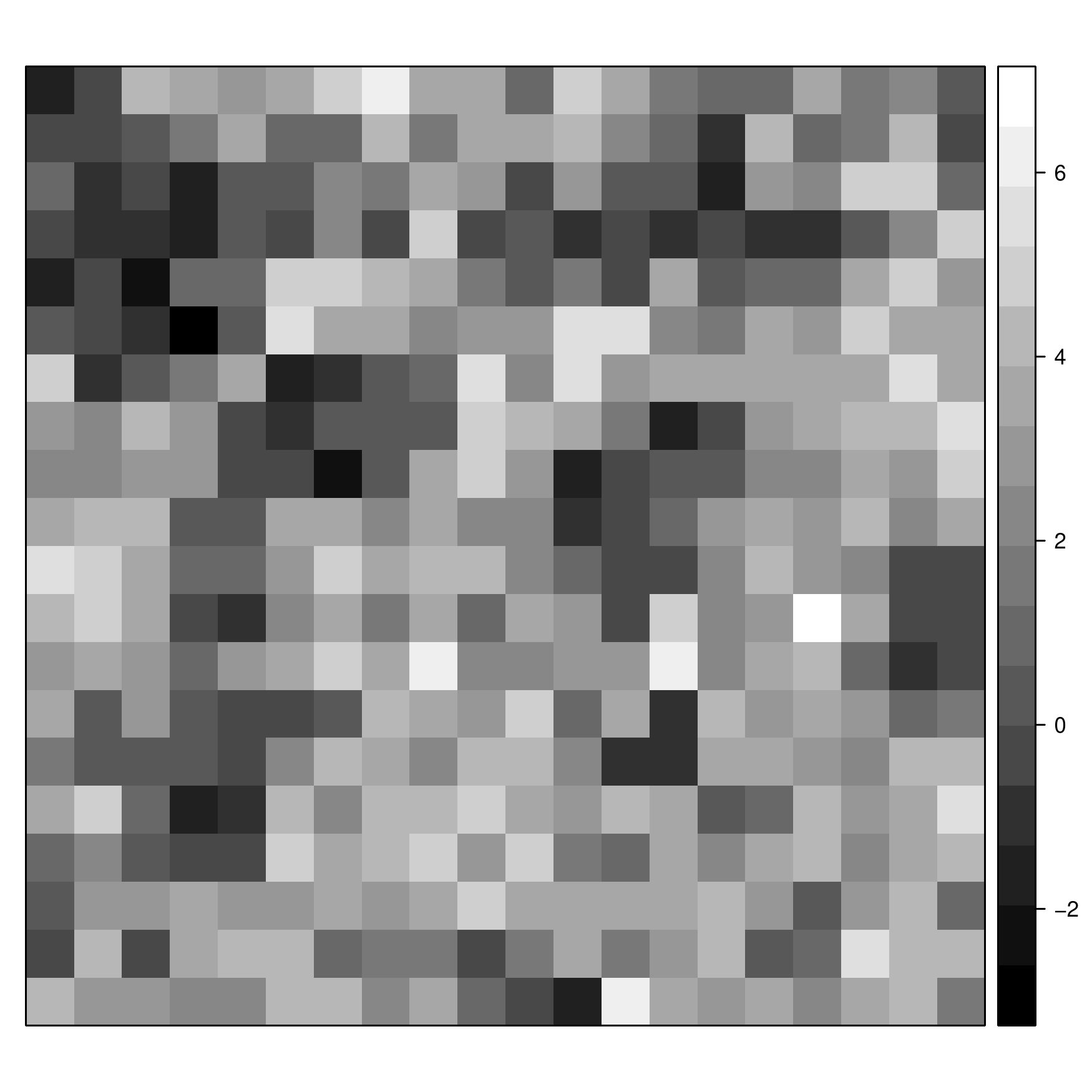}
        \label{fig:SimDatArray}
\end{figure}


\begin{figure}[htb]  
    \centering
    \caption{Neighborhood structures and weights used for the Bayesian CAR model in the simulation study. In each illustration, the center square represents the gene of interest, and the numbers are the weights $w_{ij}$ assigned to each gene in the neighborhood.}
        \includegraphics[scale= 0.6]{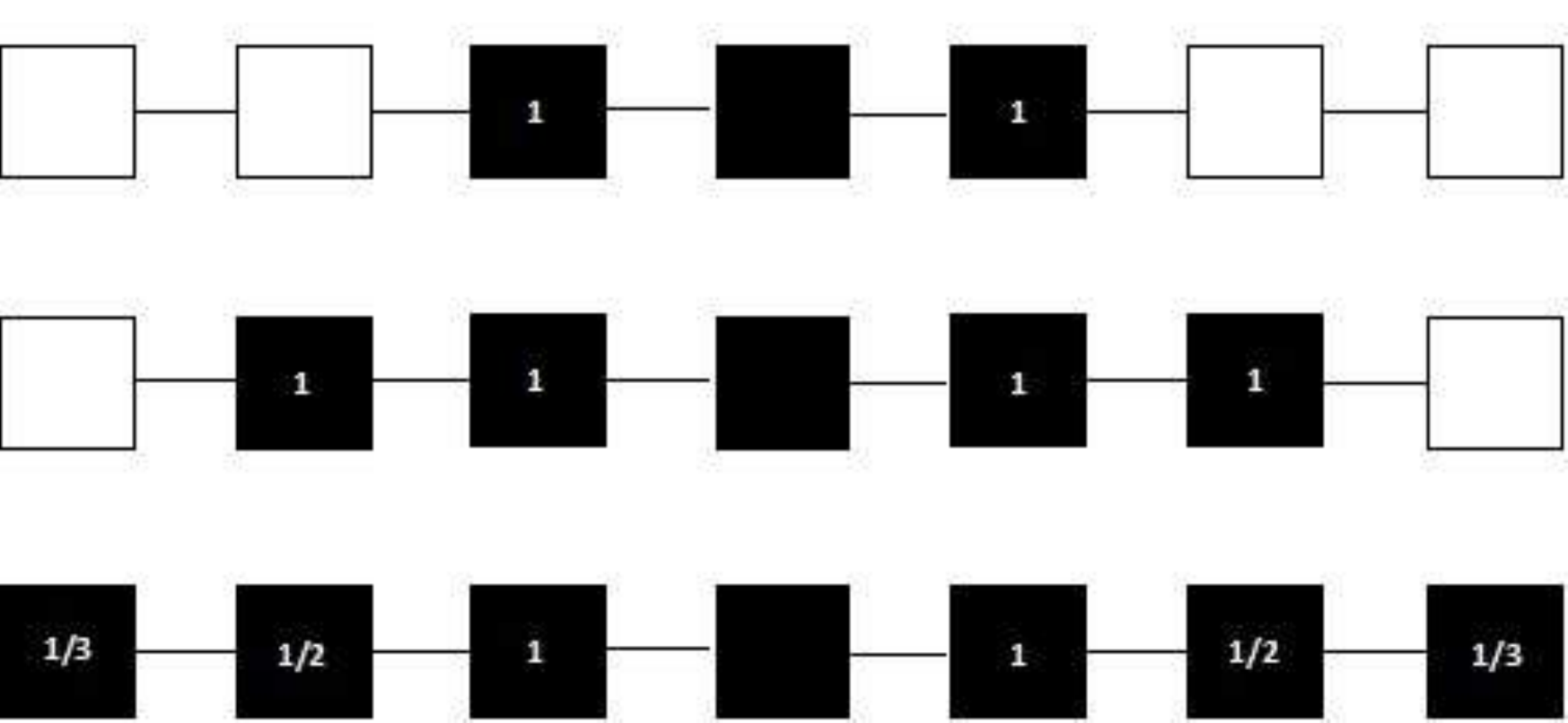}
        \label{fig:nbhds}
\end{figure}


\begin{figure}[h!]  
    \centering
    \caption{Estimated posterior inclusion probabilities $p_i$ versus test statistics $y_i$, $i= 1, \ldots, 1000$ for the simulated microarray data with neighborhoods determined through physical adjacency. The dashed (jagged) curve results from the CAR($\mathbf{W}_1$) model, and the dotted (smooth) line results from the SB model. The circled point corresponds to the indicated statistic depicted in Figure \ref{fig:pointZ}.}
        \includegraphics[scale= 0.4]{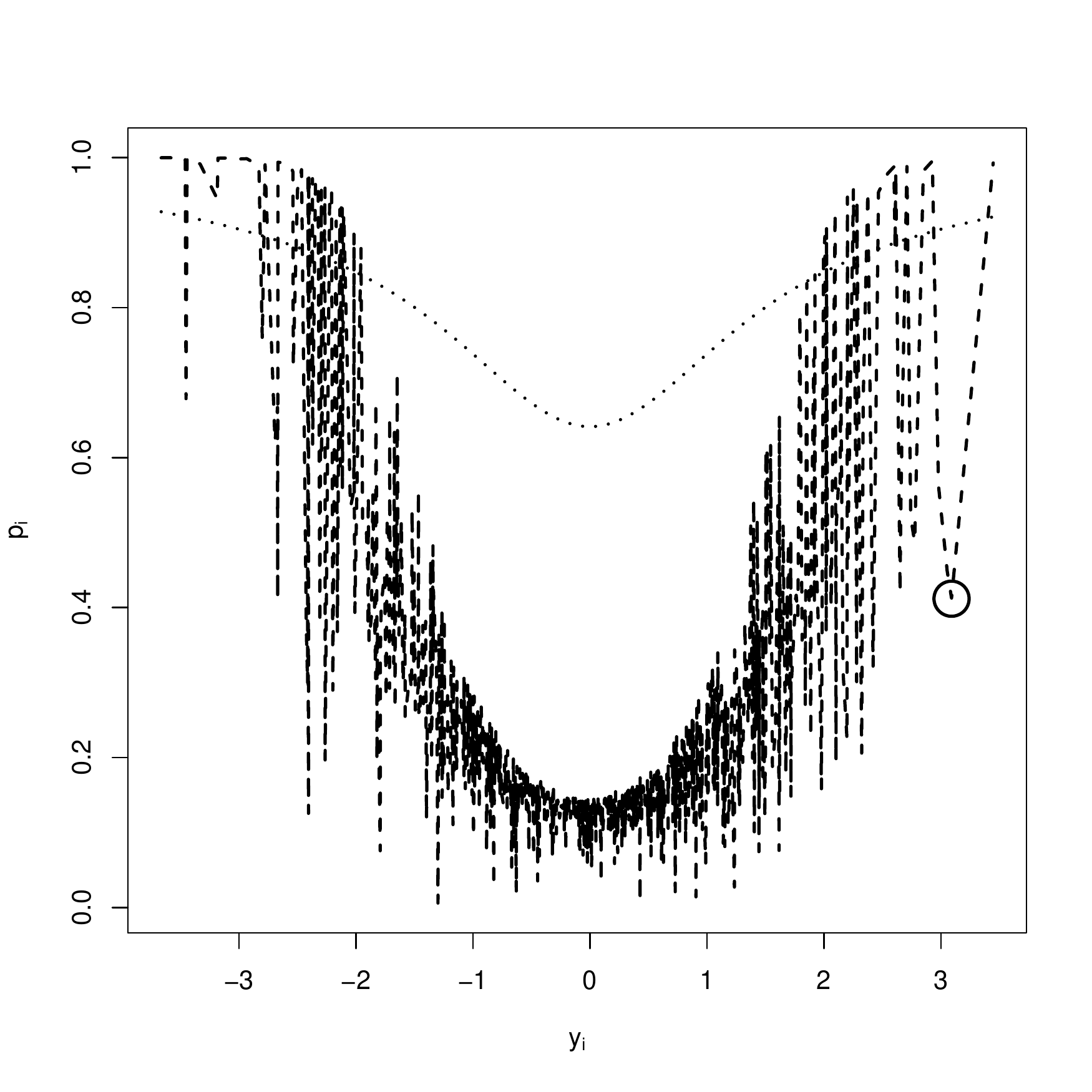}
        \label{fig:PiZi}
\end{figure}

\begin{figure}[h!]
    \begin{center}
    \caption{Graphical depiction of test statistics from the simulated microarray data. The test statistics $y_i$ are arranged in order, $i= 1, 2, \ldots, 1000$, going from the lower left to the upper right, row-wise from left to right. The circle indicates the statistic corresponding to the circled $(y_i, p_i)$ point in Figure \ref{fig:PiZi}.}
        \includegraphics[scale= 0.5, clip= true, trim= -.5inin 0in 0in 1in]{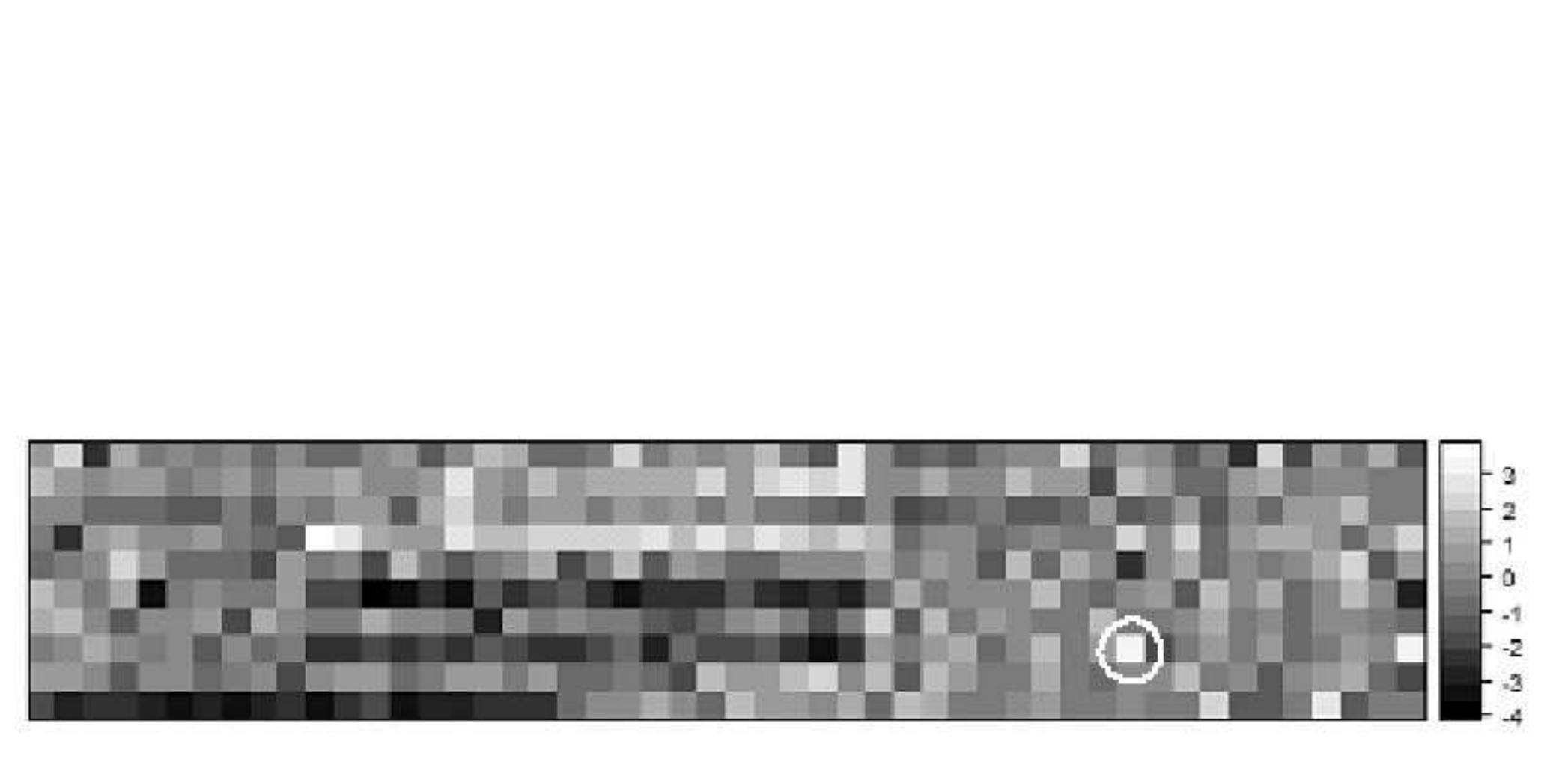}
        \label{fig:pointZ}
    \end{center}
\end{figure}

\begin{figure}[htb]  
    \centering
    \caption{Empirical ROC curves for the testing model using the physical-adjacency CAR model and the generalized CAR using pathways to define neighborhoods with isolated points included.}
        \includegraphics[scale= 0.45]{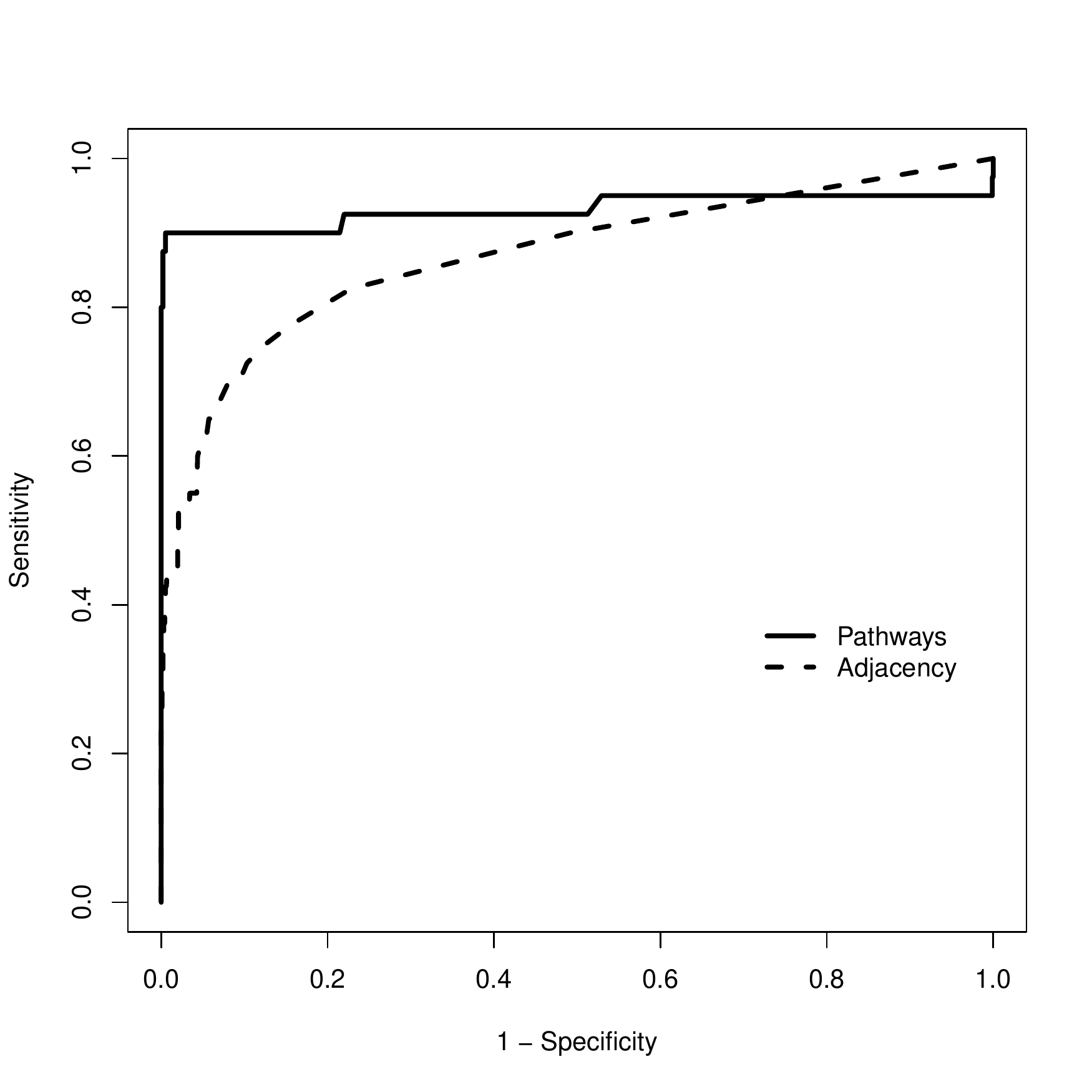}
        \label{fig:PathVsAdjROC}
\end{figure}

\begin{figure}[htb]  
    \centering
    \caption{Histogram of test statistics from the simulated pathways example with normal densities superimposed, each with mean zero and standard deviations estimated as $\sqrt{E(\sigma^2 \mid \mathbf{y})}$ from the posterior distributions of both neighborhood structures. The tick marks at the bottom indicate the non-null cases.}
        \includegraphics[scale= 0.45, trim= 0in 18pt 0in 0in]{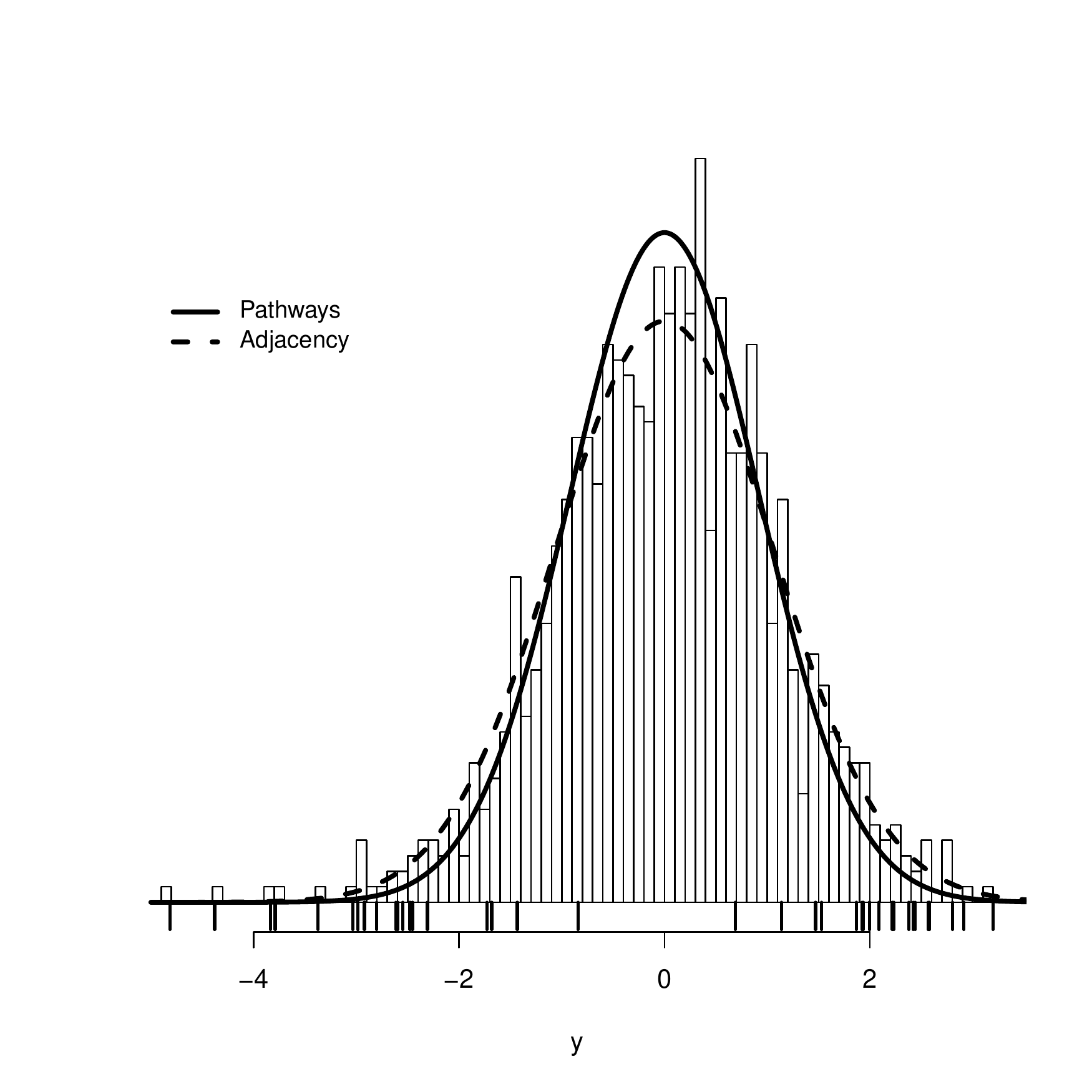}
        \label{fig:PathVsAdjHist}
\end{figure}

\begin{figure}[htb]  
    \centering
    \caption{Posterior inclusion probabilities estimated under the CAR testing model with and without the isolated cases. The dashed line is at the 0.95 threshold. Notice that several of the P5 cases have been pulled downward, resulting in more false non-discoveries by excluding the isolated cases.}
        \includegraphics[scale= 0.7, clip= TRUE, trim= 0in 0in 0in 0in]{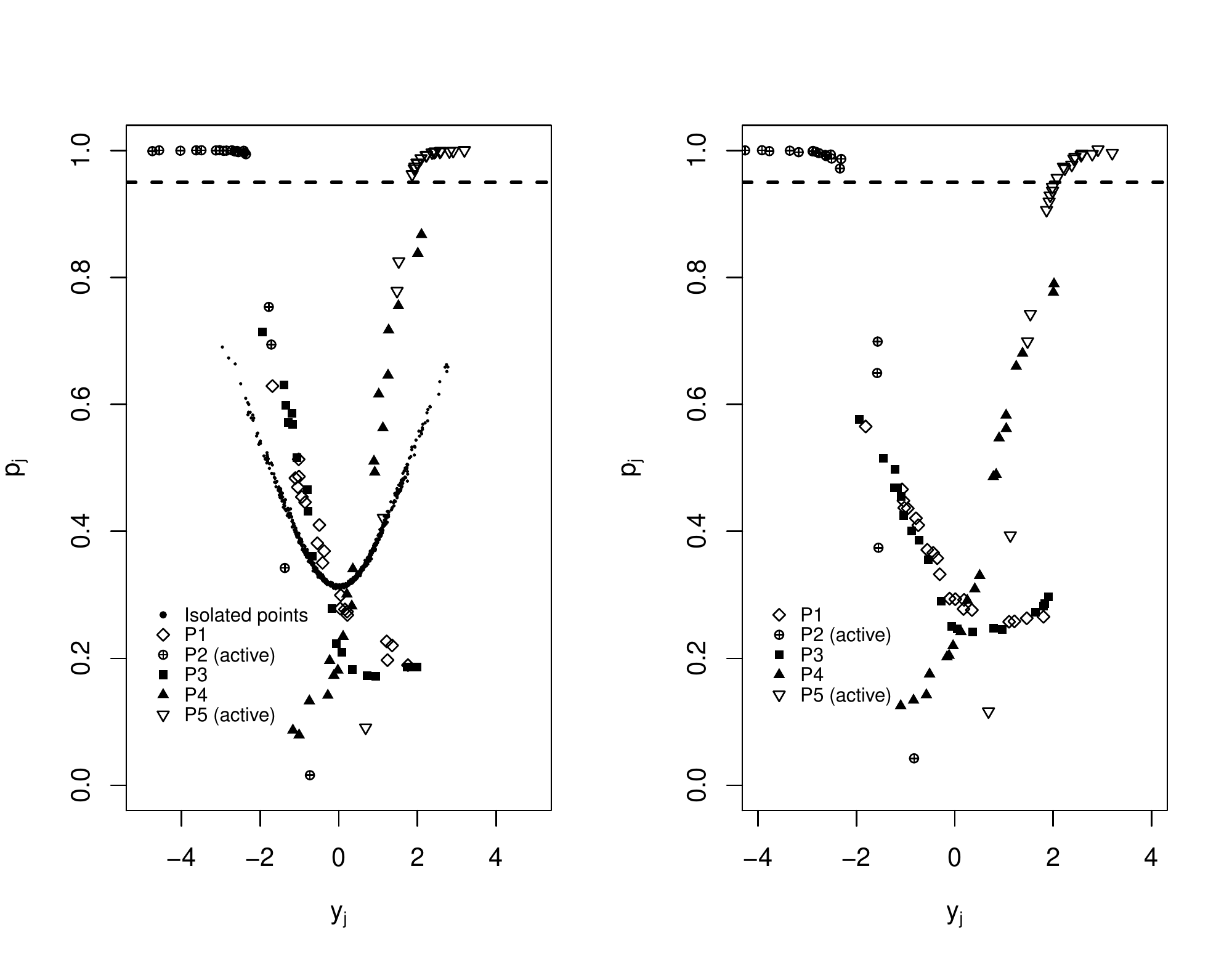}
        \label{fig:simPathFVsRProbs}
\end{figure}

\begin{figure}[htb]  
    \centering
    \caption{Histograms of realizations of Moran's $I$ calculated from the posterior predictive distribution of the CAR testing model, $p(I(\mathbf{y}^{\ast}) \mid \mathbf{y}) = \int_{\boldsymbol{\theta}}p(I(\mathbf{y}^{\ast}) \mid \boldsymbol{\theta})\pi(\boldsymbol{\theta} \mid \mathbf{y})d\boldsymbol{\theta}$, under partially incorrect correlation assumptions. The dark vertical line indicates the observed value of $I$ under the assumed neighborhood structure.}
        \includegraphics[scale= 0.7, clip= TRUE, trim= 0in 0in 0in 0in]{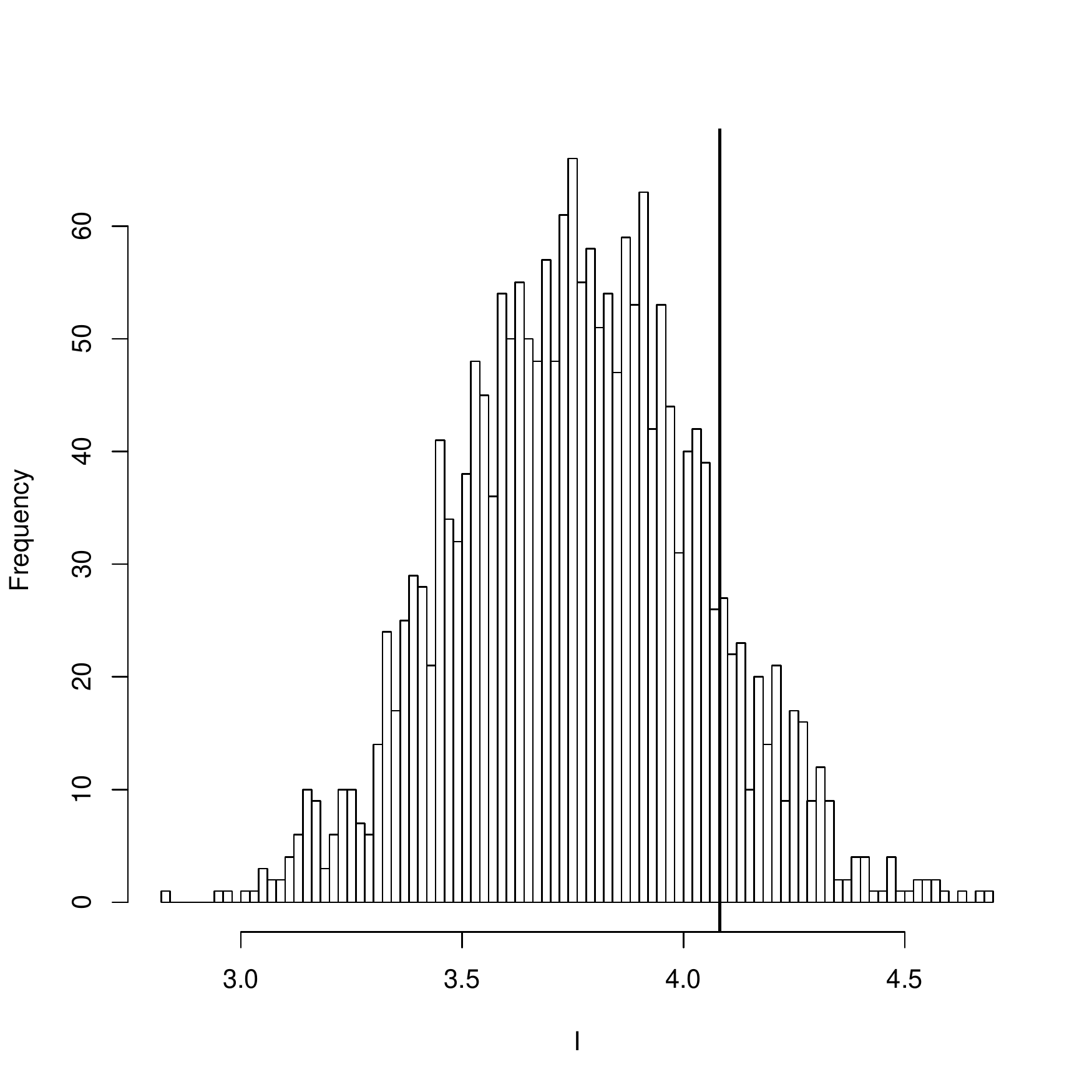}
        \label{fig:simCplxPPDMoranI}
\end{figure}

\begin{figure}[h!]  
    \centering
    \caption{Smoothed posterior densities of $p$ and $\rho$ for the E. Coli data (Xiao, Reilly, and Khodursky, 2009) with different values of $\alpha$ in the prior $p \sim \text{Beta}(\alpha, 1)$.}
        \includegraphics[scale= 0.6]{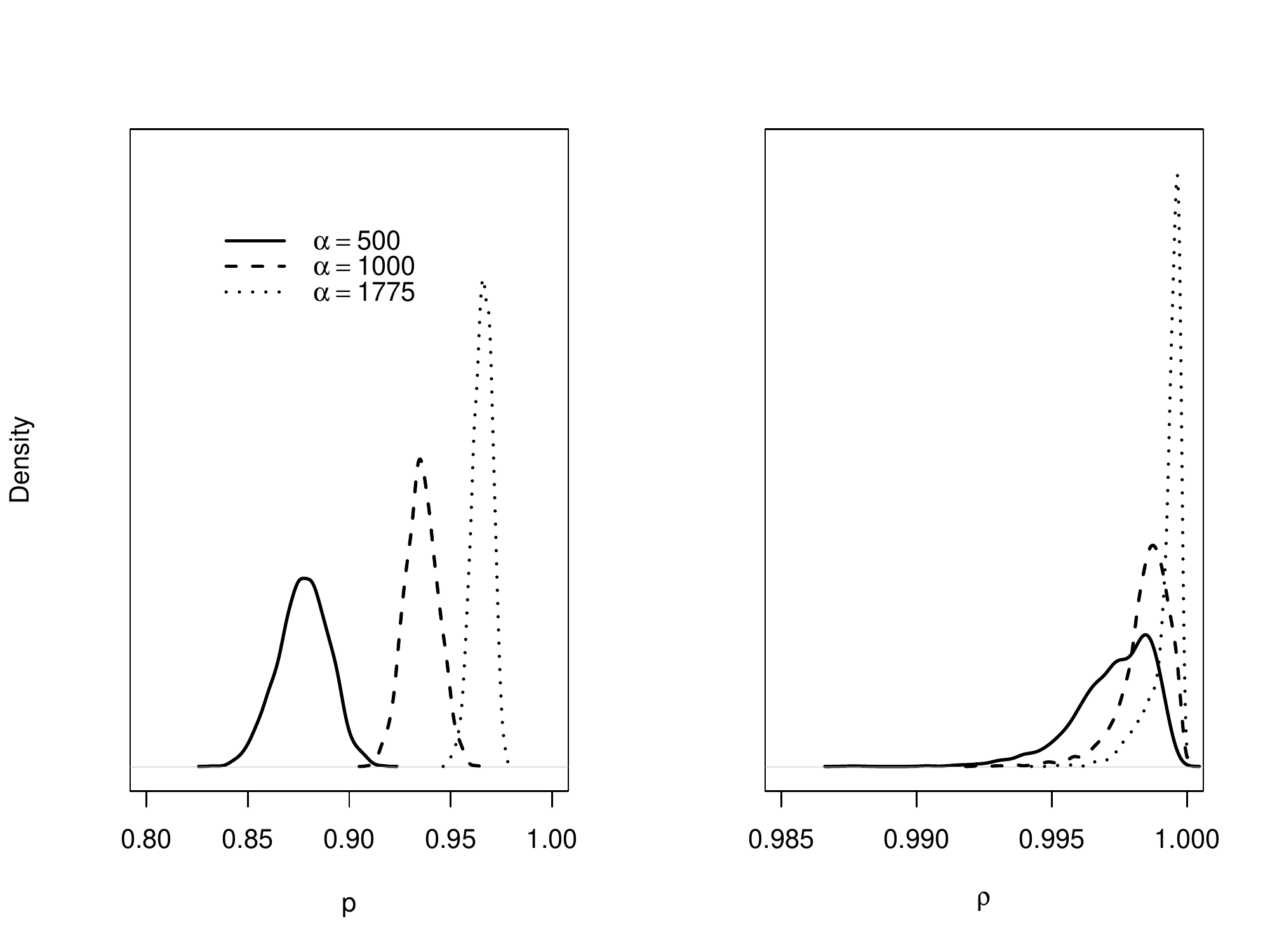}
        \label{fig:PostDensAlpha}
\end{figure}


%

\begin{figure}[h!]  
    \centering
    \caption{Estimated posterior inclusion probabilities for each of the 22,283 genes in the lymphoblastoid cell data (Subramanian, Tamayo, Mootha, Mukherjee, Ebert, Gillette, Paulovich, Pomeroy, Golub, Lander, and Meslrov, 2005) under both CAR testing models with and without isolated cases. The left panel results from excluding the isolated cases ($d = 0$), the right panel results from including isolated cases ($d = 1$). The horizontal lines represent the 0.99 threshold.}
        \includegraphics[scale= 0.6, clip= true, trim= 0in 0in 0in 0in]{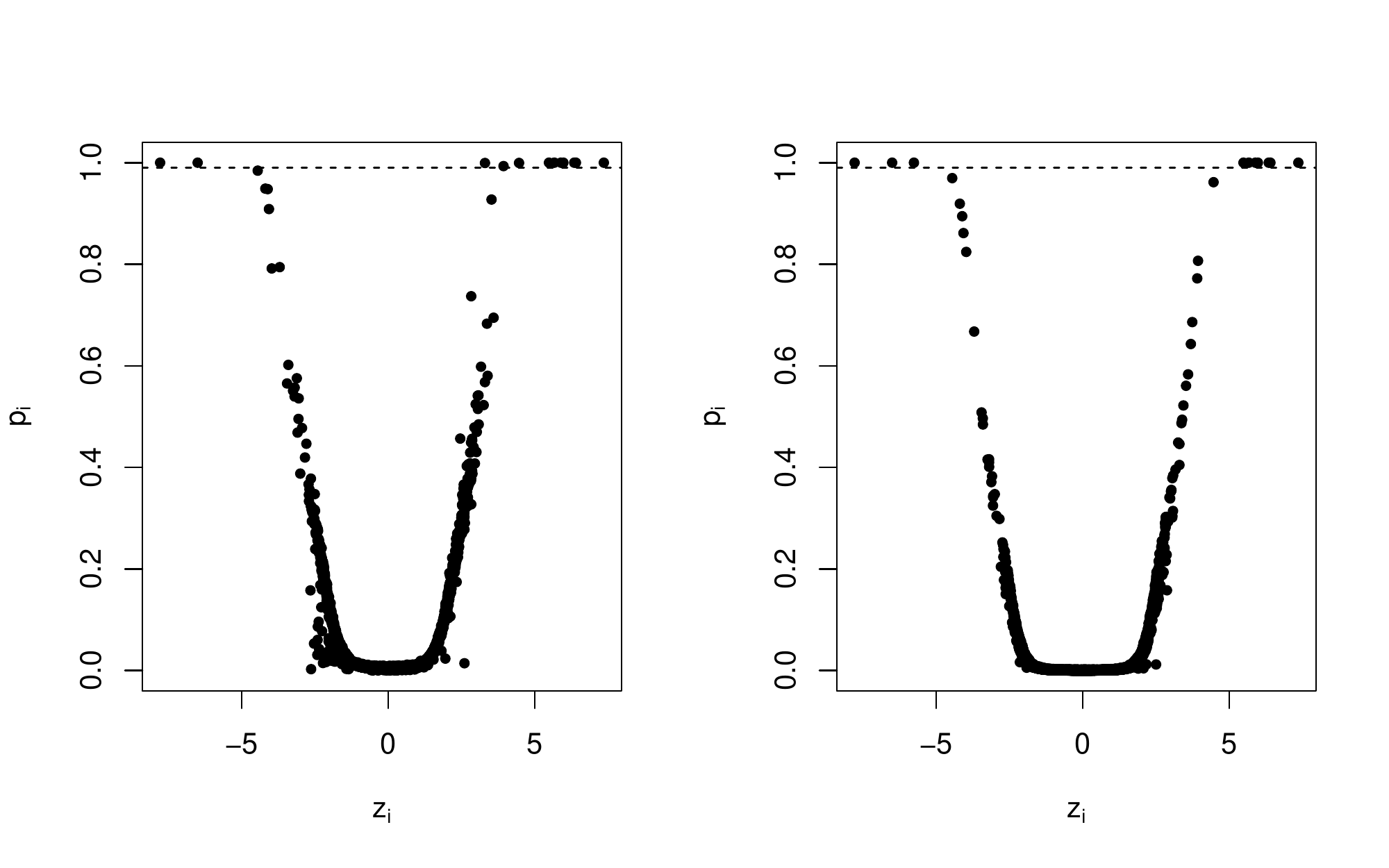}
        \label{fig:gseaGeneProbs}
\end{figure}

\end{document}